\documentclass[nofootinbib,twocolumn,pra,letterpaper,longbibliography]{revtex4-2}
\usepackage{latexsym,amsmath,amssymb,amsfonts,graphicx,color,amsthm}
\usepackage{enumerate}
\usepackage{enumitem}
\usepackage{braket}
\usepackage{adjustbox}
\usepackage{footnote}
\usepackage[dvipsnames]{xcolor}
\usepackage{tipa}
\usepackage{textcomp}
\usepackage{fontenc}
\usepackage{mathrsfs}
\usepackage{color}
\usepackage{colortbl}
\usepackage{graphics,graphicx}
\usepackage{txfonts}
\usepackage{lipsum, babel}
\usepackage{bbm}

\usepackage[unicode=true,pdfusetitle, bookmarks=true,bookmarksnumbered=false,bookmarksopen=false, breaklinks=false,pdfborder={0 0 0},backref=false,colorlinks=false] {hyperref}
\hypersetup{ colorlinks,linkcolor=myurlcolor,citecolor=myurlcolor,urlcolor=myurlcolor}

\definecolor{myurlcolor}{rgb}{0,0,0.7}

\usepackage{float}
\usepackage{hyperref}
\usepackage[T1]{fontenc}

\newcommand{\cD}{\mathcal{D}}

\newcommand{\cH}{\mathcal{H}}
\newcommand{\cI}{\mathcal{I}}

\newcommand{\cK}{\mathcal{K}}
\newcommand{\cL}{\mathcal{L}}
\newcommand{\cM}{\mathcal{M}}
\newcommand{\cN}{\mathcal{N}}

\newcommand{\cS}{\mathcal{S}}

\newcommand{\Id}{\mathbbm{1}}
\newcommand{\tr}{\text{Tr}}

\newcommand{\Proj}{\mathsf P}

\newtheorem{theorem}{Theorem}
\newtheorem{proposition}{Proposition}
\newtheorem{lemma}{Lemma}

\newtheorem{definition}{Definition}

\newtheorem{remark}{Remark}

\begin{document}

\title{Informational completeness of qubit measurements and IC preservability of qubit channels: Characterization and Quantification}

\author{Jatin Ghai$^{1,2,}$}
\email{jghai98@gmail.com}

\author{Arindam Mitra$^{3,}$}
\email{am56@iitbbs.ac.in}
\email{arindammitra143@gmail.com}

\affiliation{$^1$Optics and Quantum Information Group, The Institute of Mathematical Sciences, C. I. T. Campus, Taramani, Chennai 600113, India.\\
$^2$Homi Bhabha National Institute, Training School Complex, Anushaktinagar, Mumbai 400094, India.\\
$^3$Department of Physics, School of Basic Sciences, Indian Institute of Technology Bhubaneswar, Odisha 752050, India.
}

\begin{abstract}
    Informationally complete (IC) measurements are a useful class of measurements, as their outcome statistics uniquely determine an unknown quantum state. Hence, they are important for certain tasks such as quantum state tomography, quantum process tomography, etc. In this work, we study the quantification of informational completeness for arbitrary quantum measurements by introducing and characterizing a faithful measure for it. We explicitly evaluate the informational completeness of qubit symmetric informationally complete (SIC) measurements and show that it is an upper bound for all qubit minimal informationally complete measurements. Furthermore, by introducing a faithful measure, we try to quantify and characterize the ability of an arbitrary quantum channel to preserve informational completeness of any IC measurement when the channel acts on it in the Heisenberg picture. We call this measure informational completeness-preservability (IC preservability) of quantum channels. After studying its properties, we establish its relation to another quantity, namely, the absolute output coherence of a quantum channel, which quantifies the minimum amount of coherence (w.r.t. an arbitrary incoherent basis) that can always be obtained from the output of that channel. Finally, we illustrate our results with an example from light-matter interaction scenarios. Thus, in this work, not only do we try to provide a quantitative framework for studying both the informational completeness of quantum measurements and the ability of quantum channels to preserve it, but we also try to offer key insight into the conceptual relation between informational completeness and quantum coherence. 
\end{abstract}

\maketitle

\section{Introduction}

Quantum measurements play a central role in quantum mechanics. They serve as a bridge between the abstract mathematical description of a quantum system and the actual information about the quantum system that can be accessed in practical experiments. Thus, the characterization of quantum systems is fundamentally based on the extraction of information through measurement. In many quantum information-processing tasks, such as quantum state tomography and quantum process tomography, it is necessary to perform measurements whose outcome statistics are sufficient to uniquely determine a completely unknown quantum state. This class of measurements is known as \emph{informational complete} (IC) measurements\cite{Ariano_IC_2004,busch_1991_informationally,Ariano_2005_IC}. For an unknown quantum state, the statistics generated by an informationally complete measurement contain sufficient information to uniquely determine that unknown state. Consequently, informationally complete measurements are necessary for both the foundational exploration of the quantum theory, as well as, for the practical side of quantum theory in experimental settings. Among informationally complete measurements, a special class of measurements is \emph{symmetric informationally complete} (SIC) measurements, also known as SIC-POVMs (Symmetric Informationally Complete Positive Operator-Valued Measures)\cite{Renes_SIC_Conj,samuel_2025_symmetric}. A SIC measurement in a $d$-dimensional Hilbert space consists of $d^2$ rank-one measurement operators with equal pairwise overlaps. Besides the convenient symmetrical structure, SIC-POVMS also have practical importance. They are very useful for quantum-state tomography\cite{Wooters_1989,Scott_2006}, and has applications in entanglement detection\cite{Shang_2018}, dimension witnessing\cite{Brunner_2013}, device-independent randomness generation\cite{Andersson_2018}, and the certification of quantum measurement devices\cite{Tavakoli_2020,Miron_2019}.

If white noise is mixed with an informationally complete measurement, it may remain informationally complete\cite{Mitra_2026}. But the addition of white noise makes it less robust. In other words, its ability to discriminate between quantum states is decreased. Therefore, a natural question is how to quantify the informational completeness of a given quantum measurement, and such an investigation has not been widely conducted yet to the best of our knowledge. Since neither informationally complete nor incomplete measurements form a convex set\cite{Heinosaari_book_QF,watrous_2018_theory}, measures such as robustness\cite{Roope_2015_Rob,Paul_2019_Rob} and weight\cite{Roope_2020_Weight} are generally not well-defined or well-behaved. Therefore, in this work, we are \emph{motivated} to quantify the informational completeness of quantum measurements. Now, when a quantum channel acts on quantum measurements in the Heisenberg picture, it may decrease the quality of informational completeness or destroy it. Therefore, it is important to study the informational completeness preservability of quantum channels, and such an investigation also serves as another \emph{motivation} behind this work. In this work, we also demonstrate the relation of informational completeness preservability with absolute output coherence of a quantum channel, which offers key insight into the conceptual relationship between informational completeness and quantum coherence. It should also be mentioned that apart from a few general results, we mostly restrict ourselves to qubit systems in this work.

The rest of the work is organized as follows. In Sec. \ref{Sec:Preliminaries} we discuss the necessary preliminaries. We discuss our main results in Sec. \ref{Sec:Main}. More specifically, in Sec. \ref{Sec:Info_comp} we study the quantification of \emph{information completeness} for quantum measurements and explore its properties. In Sec. \ref{Sec:Info_Comp_Pres}, we quantify the \emph{informational completeness preservability} of quantum channels and try to characterize it. In Sec. \ref{Sec:Coh_Pres_Pow}, we define and study the \emph{absolute output coherence} of a qubit channel and establish its relation with IC preservability of the same qubit channel (see Theorem \ref{Th:IC_Coh_rel}). In Sec. \ref{Sec:illustration}, we illustrate our results with an example from light-matter interaction scenarios. In Sec. \ref{Sec:Conclusion}, we summarize our results and discuss future directions.

\section{Preliminaries}\label{Sec:Preliminaries}

\subsection{Quantum Measurements and Channels}

A quantum measurement, or positive operator-valued measure (POVM), on a Hilbert space $\cH$ is described by a set of positive semidefinite operators
$A=\{A(x)\in\cL(\cH)\}_{x\in\Omega_A}$, satisfying the normalization condition $\sum_{x\in\Omega_M}A(x)=\mathbbm{1}_{\cH}$, where $\mathbbm{1}_{\cH}$ denotes the identity operator on $\cH$\cite{Heinosaari_book_QF}. The set $\Omega_A$ denotes the possible outcomes of the measurement, while $\cL(\cH)$ corresponds to the space of linear operators on $\cH$.  In this work, we restrict ourselves to finite-dimensional Hilbert spaces. Each operator $A(x)$ is referred to as a POVM element of the measurement $A$. A measurement $A$ is called projective if every POVM element satisfies $A(x)^2=A(x),~\forall~x\in\Omega_A$. In this work, only measurements with finitely many outcomes are considered. When a measurement $A$ is performed on a quantum state $\rho\in\cS(\cH)$, where $\cS(\cH)$ denotes the set of density operators on $\cH$, the probability of observing outcome $x$ is given by the Born rule $p(x)=\tr[\rho A(x)]$. We denote the set of all measurements on $\cH$ by $\mathscr{M}(\cH)$. For qubits specifically, we denote their Hilbert space as $\cH_Q$.

Quantum channels describe generic physically realisable deterministic transformations between quantum states. Mathematically, a quantum channel $\Lambda:\cL(\cH)\rightarrow\cL(\cK)$ is represented by a completely positive and trace-preserving (CPTP) linear map~\cite{Heinosaari_book_QF}. In the Schrödinger picture, the channel acts on density operators, mapping states on the Hilbert space $\cH$ to states on the Hilbert space $\cK$. With any linear map
$\Lambda:\cL(\cH)\rightarrow\cL(\cK)$, there exists
another associated map, known as its dual map $\Lambda^\dagger:\cL(\cK)\rightarrow\cL(\cH)$,
defined through the relation
$\tr[Y^\dagger\Lambda^\dagger(X)]=\tr[(\Lambda(Y))^\dagger X]$, for all $X\in\cL(\cK)$ and $Y\in\cL(\cH)$\cite{watrous_2018_theory}. The dual map corresponds to the description of the quantum channel in the Heisenberg picture and in this picture it transforms the measurements in $\cM(\cK)$ to the measurements in $\cM(\cH)$. More specifically, if $A=\{A(x)\in\cL(\cK)\}$ is a measurement on the output space of a channel $\Lambda$, then the transformed measurement on the input space in the Heisenberg picture is given by $\Lambda^\dagger(A):=\{\Lambda^\dagger(A(x))\}$. The set of all quantum channels from $\cL(\cH)$ to $\cL(\cK)$ is denoted by $\mathscr{C}(\cH,\cK)$. Given two linear maps $\Lambda_1:\cL(\cH_2)\rightarrow\cL(\cH_1)$ and 
$\Lambda_2:\cL(\cH)\rightarrow\cL(\cH_2)$,
their composition is written as $\Lambda_1\circ\Lambda_2$ and is defined by $(\Lambda_1\circ\Lambda_2)(\rho)
:=
\Lambda_1(\Lambda_2(\rho)),
~\forall~\rho\in\cL(\cH)$. If both maps are quantum channels, then $\Lambda_1\circ\Lambda_2\in\mathscr{C}(\cH,\cH_1)$. For two channels $\Lambda\in\mathscr{C}(\cH,\cK)$ and $\Lambda^\prime\in\mathscr{C}(\cH,\cK^\prime)$, we say that $\Lambda$ is a post-processing of $\Lambda^\prime$, if there exists a channel $\Theta\in\mathscr{C}(\cK^\prime,\cK)$ such that $\Lambda=\Theta\circ\Lambda^\prime$. In this case, we write $\Lambda^\prime\succeq_{postproc}\Lambda$. The relation $\succeq_{postproc}$ defines a preorder on quantum channels and is commonly referred to as the \emph{post-processing preorder}\cite{Heinosaari_incomp_chan}. A linear map $\Lambda:\cL(\cH)\rightarrow\cL(\cK)$ is said to be Hermitian-preserving trace preserving (HPTP) linear map if for arbitrary $X\in Herm(\cH)$ we have $\Lambda(X)\in Herm(\cK)$ and $\tr[\Lambda(X)]=\tr[X]$ \cite{watrous_2018_theory}. Here, $Herm(\cH)$ is a set of Hermitian matrices in $\cL(\cH)$. In the case, when $\Theta$ is just an HPTP map in the relation $\Lambda=\Theta\circ\Lambda^\prime$ the preorder is denoted as $\Lambda^\prime\succeq_{asymp}\Lambda$\cite{Mitra_2026}. Given a quantum channel $\Lambda$, a linear map $\mathcal{N}:\cL(\cK)\rightarrow\cL(\cK^{\prime})$ is  called a statistical morphism if for any POVM $A=\{A(x)\}\in\mathscr{M}(\cK^{\prime})$, there exists another POVM $A_{\cN}=\{A_{\cN}(x)\in\mathscr{M}(\cK)\}$ such that~\cite{Buscemi_2018_rev_data, Buscemi_2012_comparison,Buscemi_2016_degradable} for all $\rho\in\cS(\cH)$
    \begin{align}
    \tr[\mathcal{N}\circ\Lambda(\rho)A(x)]
        =
        \tr[\Lambda(\rho)A_{\cN}(x)].\label{Eq:Stat_morph}
    \end{align}

\subsection{Informationally Complete Measurements}

Informationally complete (IC) measurements are an important class of quantum measurements\cite{Ariano_IC_2004,busch_1991_informationally,Ariano_2005_IC}. Their defining feature is that for an arbitrary quantum state, the measurement statistics is sufficient to uniquely determine that underlying quantum state. In other words, for an informationally complete measurement $A=\{A(x)\}$, if  $\tr[\rho A(x)]=\tr[\omega A(x)]~\forall x\in \Omega_A$ hold for a pair of quantum states $\rho,\omega\in\cS(\cH)$ then $\rho=\omega$ must hold by definition. For a quantum system described by a $d$-dimensional Hilbert space $\cH$, an IC measurement requires at least $d^2$ POVM elements. Formally, a POVM $A=\{A(x)\}_{x=1}^N$ where $N\geq d^2$, is informationally complete if its elements span the full operator space $\cL(\cH)$. Consequently, any operator $X\in\cL(\cH)$ has a decomposition of the form
\begin{equation}
    X=\sum_{x=1}^{N} c_x A(x).
    \label{Eq:IC_POVM_Exp}
\end{equation}
 In general, all coefficients $\{c_x\}$ need not be non-negative, even when $X$ itself is positive semidefinite. In general, the POVM elements $\{A(x)\}$ of an IC measurement form an overcomplete basis of $\cL(\cH)$, implying that the decomposition in Eq.~\eqref{Eq:IC_POVM_Exp} is not unique. If a measurement is not informationally complete, it is informationally incomplete. A particularly important subclass of IC measurements is \emph{minimal informationally complete measurements}, for which the number of POVM elements is exactly $d^2$. In this case, the POVM elements are linearly independent and form a basis for $\cL(\cH)$, ensuring that the expansion in Eq.~\eqref{Eq:IC_POVM_Exp} is unique. Minimal informationally complete measurements are useful for describing the dynamics of finite dimensional quantum systems using probability distribution\cite{yashin_2020_minimal}.

\subsection{Symmetric informationally complete measurements}

Within minimal informationally complete measurements, another important subclass is formed by \emph{symmetric informationally complete} (SIC) POVMs. A SIC POVM $A$ in dimension $d$ consists of $d^2$ subnormalized rank-one operators \cite{Renes_SIC_Conj} $A(x)=\frac{1}{d}\ket{\psi_x}\bra{\psi_x}, \forall x\in\{1,\ldots,d^2\}$, satisfying the symmetry condition $|\langle \psi_x | \psi_y \rangle|^2=\frac{1}{d+1},
\qquad \forall~x\neq y$, together with the normalisation relation $\sum_x A(x)=\mathbbm{1}_{\cH}$. Since a SIC POVM contains exactly $d^2$ linearly independent POVM elements, it is minimal by construction. SIC POVMs are conjectured to exist in every finite dimension\cite{Renes_SIC_Conj,samuel_2025_symmetric}. Explicit analytical constructions are known for several dimensions\cite{Appleby_SIC_2017,Appleby_SIC_2018,Appleby_SIC_2022,Bengtsson_SIC_2025}, while numerical solutions for higher-dimensional cases have also been reported in the literature\cite{Appleby_SIC_2018,Scott_SIC_num}.

For qubit systems ($d=2$), a SIC POVM consists of four subnormalized rank-one operators $M(x)=\frac{1}{2}\ket{\psi_x}\bra{\psi_x}, \forall x=\{1,2,3,4\}$, where the corresponding pure states form the vertices of a regular tetrahedron on the Bloch sphere. In terms of Bloch vectors, an SIC measurement can be expressed as
\begin{align}
    A(x)=\frac{1}{4}(\mathbbm{1}+\textbf{r}(x).\boldsymbol{\sigma}),\label{Eq:SIC_Bloch_rep}
\end{align}
where the Bloch vectors $\textbf{r}(x)\in\mathbbm{R}^3$ for $x=1,2,3,4$ satisfy $\textbf{r}(x).\textbf{r}(y)=-\frac{1}{3}$. Furthermore, since $\sum_x A(x)=\mathbbm{1}$, we obtain $\sum_x \textbf{r}(x)=0$.

One convenient choice of such Bloch vectors following the vertices of a regular tetrahedron is the following:
\begin{align}
    &\textbf{r}(1)=\frac{1}{\sqrt{3}}(1,1,1),\qquad
    \textbf{r}(2)=\frac{1}{\sqrt{3}}(1,-1,-1),\label{Eq:SIC_vec12}\\
    &\textbf{r}(3)=\frac{1}{\sqrt{3}}(-1,1,-1),\qquad
    \textbf{r}(4)=\frac{1}{\sqrt{3}}(-1,-1,1).\label{Eq:Sic_vec34}
\end{align}



\subsection{Qubit Channel Parametrization}

In the Bloch sphere picture, a qubit density matrix can be written as $\rho=\frac{1}{2}(\mathbbm{1}+\textbf{r}.\boldsymbol{\sigma})\in\cS(\cH_Q)$, where $\textbf{r}=(r_1,r_2,r_3)$ is the Bloch vector of the quantum state $\rho$ and it satisfies $|\textbf{r}|\leq 1$. Also, $\boldsymbol{\sigma}=(\sigma_1,\sigma_2,\sigma_3)$ is the vector of three Pauli matrices.  One can define a $1\times 4$ vector $\boldsymbol{\nu}_{\rho}$  associated with the quantum state $\rho$ such that $(\boldsymbol{\nu}_{\rho}):=(1,\textbf{r})$. Then, an arbitrary linear map $\Lambda:\cL(\cH)\rightarrow\cL(\cK)$ acting on qubit states can be represented by a real $4\times 4$ matrix $\mathbf{T}_\Lambda$ through
\begin{align}
\boldsymbol{\nu}_{\Lambda(\rho)}=\mathbf{T}_\Lambda (\boldsymbol{\nu}_{\rho})=(\mathbf{T}_\Lambda \boldsymbol{\nu}_{\rho}^T)^T,
\end{align}

where $T$ denotes the transpose. If $\Lambda$ is a CPTP map, this matrix has the block form\cite{braun_2014_universal,Ruskai_2001,Ruskai_2002}
\begin{align}
\textbf{T}_\Lambda=
\begin{pmatrix}
1 & 0_{1\times 3}\\[2mm]
\textbf{t}^T_{\Lambda} & M_\Lambda
\end{pmatrix},
\end{align}
where $M_\Lambda\in \mathbb{R}^{3\times 3}$, $\textbf{t}_\Lambda=(t_1,t_2,t_3)\in\mathbb{R}^3$.  
Equivalently, the action on the Bloch vector can be written as:
\begin{align}
\textbf{r}^{\prime} = M_\Lambda (\textbf{r}) + \textbf{t}_\Lambda.
\end{align}
where $\textbf{r}^{\prime}$ is the block vector of $\Lambda(\rho)$.
For any given matrix $M_\Lambda$, there exist two orthogonal matrices $O_1$ and $O_2$ such that singular value decomposition of $M_{\Lambda}$ can be written as
\begin{align}
    M_{\Lambda}=O_1M_{\Lambda_D}O_2.
\end{align}
For any $3\times3$ orthogonal matrix $O$, either $O$ (proper rotation) or $-O$ (improper rotation) is in $SO(3)$. Accordingly, $M_{\Lambda}=O_1M_{\Lambda_D}O_2$ with $M_{\Lambda_D}=D$ if both $O_1$ and $O_2$ are in $SO(3)$, or, $M_{\Lambda_D}=-D$ if only one of $O_1$ and $O_2$ is in $SO(3)$ where
\begin{align}
D=
\begin{pmatrix}
\lambda_1 & 0 & 0\\
0 & \lambda_2 & 0\\
0 & 0 & \lambda_3
\end{pmatrix}.
\end{align}
Here  $\lambda_1,\lambda_2,\lambda_3$ are the signed singular values of $\Lambda$. Since every rotation in $SO(3)$ corresponds to a unitary transformation on the qubit Hilbert space, the channel $\Lambda$ can be expressed as
\begin{align}
\Lambda=\mathcal{U}\circ\Lambda_{D}\circ\mathcal{V},
\end{align}
where $\mathcal{U}(\rho)=U\rho U^\dagger$, and $\mathcal{V}(\rho)=V\rho V^\dagger$, for suitable unitary operators $U$ and $V$, and $\Lambda_{D}$ is the quantum channel with diagonal $M_{\Lambda_D}$.

\subsection{Characterization of Quantum Coherence}

For a fixed reference basis $\{\ket{i}\}$ of a Hilbert space $\cH$, a quantum state is said to contain no \emph{coherence} if it is diagonal in that basis. Explicitly, incoherent states are of the form $\omega=\sum_i p_i \ket{i}\bra{i}$, where $p_i\geq0$ and $\sum_i p_i=1$ and the reference basis $\{\ket{i}\}$ is known as incoherent basis\cite{Stretslov_2015,Baum_2014}. Any quantum state that cannot be expressed in this form possesses a non-zero amount of coherence.

It is important to emphasize that quantum coherence is basis-dependent. A state that is incoherent with respect to one reference basis may become coherent when represented in a different basis.

If $\mathcal{I}$ denotes the set of all incoherent density operators, then a CPTP map $\Lambda$ is called an incoherent operation if it has Kraus decomposition
\begin{equation}
\Lambda(\rho)=\sum_i K_i \rho K_i^\dagger,
\end{equation}
such that,
\begin{equation}
\frac{K_i\rho K_i^\dagger}{\tr[K_i\rho K_i^\dagger]} \in \mathcal{I},\quad \forall \rho\in\cI.
\end{equation}

Any valid measure of coherence $\mathcal{C}(\rho)$, under the standard axiomatic approach to resource theories, is postulated to satisfy the following properties\cite{Stretslov_2015,Stretslov_2017_Coh_Res}:     (1) $\mathcal{C}(\rho)\ge 0$, with $\mathcal{C}(\rho)=0$ iff $\rho\in\mathcal{I}$, (2) $\mathcal{C}(\rho)\geq \mathcal{C}(\Lambda(\rho))$ for every incoherent operation  $\Lambda$, (3) $\mathcal{C}(\rho)\geq \sum_i p_i\, \mathcal{C}(\tilde{\rho}_i)$ where $p_i=\operatorname{Tr}\left(K_i\rho K_i^\dagger\right)$ with
$\tilde{\rho}_i=\frac{K_i\rho K_i^\dagger}{p_i}$, and the Kraus operators satisfy $\frac{K_i\omega K_i^\dagger}{\tr[K_i\omega K_i^\dagger]} \in \mathcal{I} \forall \omega\in\cI.$, (4) $\mathcal{C}(\sum_i p_i\rho_i)\le \sum_i p_i \mathcal{C}(\rho_i)$.

Any functional satisfying the above properties is referred to as a \emph{coherence monotone}. Several such measures have been proposed in the literature\cite{Stretslov_2017_Coh_Res}. One of such measures is the family of distance-based coherence measures, defined as
\begin{align}
    \mathcal{C}(\rho)
    =
    \min_{\omega\in\mathcal{I}}
    \mathcal{D}(\rho,\omega),\label{Eq:distance_based_coherence_measure}
\end{align}
where $\mathcal{D}$ is a suitable distance measure between quantum states.

Generally, the distance $\mathcal{D}$ is required to be contractive under completely positive trace-preserving (CPTP) maps and to satisfy the joint convexity property. In other words, for any quantum channel $\Lambda$, and for arbitrary quantum states $\rho,\rho^{\prime},\omega,\omega^{\prime}\in\cS(\cH)$, $\mathcal{D}(\Lambda(\rho),\Lambda(\omega))\leq\mathcal{D}(\rho,\omega)$ and $p\mathcal{D}(\rho,\omega)+(1-p)\mathcal{D}(\rho^{\prime},\omega^{\prime})\geq\mathcal{D}(p\rho+(1-p)\rho^{\prime},p\omega+(1-p)\omega^{\prime})$ must hold.

This condition ensures that $C$ satisfies the above-mentioned properties.

\subsection{Optical Master Equation}

The interaction of a quantum system with the environment results in dissipation and decoherence. In the open system framework, after tracing out the environmental degrees of freedom, the system evolution is generally non-unitary. Such dynamics are typically irreversible and are governed by master equations.  The solution of such a master equation is called a dynamical map which is a set of quantum channels parametrized by time.

Both the quantum system and the environment evolve under the total Hamiltonian of the form
\begin{equation}
H_{\mathrm{tot}} = H_S + H_B + H_I,
\end{equation}
where $H_S$ is the system Hamiltonian, $H_B$ is the bath Hamiltonian, and $H_I$ is the interaction Hamiltonian. The interaction is of the form,
\begin{equation}
H_I = \sum_\alpha A_\alpha \otimes B_\alpha,
\end{equation}
where $A_\alpha$ are system operators and $B_\alpha$ are bath operators.
.

For a system weakly coupled to a large bath, under Born, Markov, and rotating-wave approximations, the evolution of the system is governed by the Gorini--Kossakowski--Sudarshan--Lindblad (GKSL) equation of the form\cite{breuer_2002_theory}
\begin{equation}
\frac{d\rho}{dt}=-i[H_{\mathrm{eff}},\rho]+\sum_k \gamma_k\left(L_k \rho L_k^\dagger-\frac12\{L_k^\dagger L_k,\rho\}\right),
\label{Eq:GKSL}
\end{equation}
where $H_{\mathrm{eff}}$ is an effective Hamiltonian\cite{breuer_2002_theory}, $L_k$ are Lindblad (or jump) operators, $\gamma_k\ge 0$ are decay rates, and $\{X,Y\}=XY+YX$ denotes the anticommutator.

A common model for the bath is a continuum of harmonic oscillators,
\begin{equation}
H_B = \sum_\mu \omega_\mu b_\mu^\dagger b_\mu,
\end{equation}
where $b_\mu$ and $b_\mu^\dagger$ are bosonic annihilation and creation operators satisfying
\begin{equation}
[b_\mu,b_\nu^\dagger]=\delta_{\mu\nu}, \qquad [b_\mu,b_\nu]=[b_\mu^\dagger,b_\nu^\dagger]=0.
\end{equation}

In quantum optical systems, a system is considered to be an atom or a molecule, while the environment is typically modeled as a continuum of electromagnetic field modes. The interaction between the system and the environment is given by dipole interaction as\cite{breuer_2002_theory}
\begin{align}
    H_I=-\mathbf{D}.\mathbf{E},
\end{align}
where $\mathbf{D}$ is the dipole operator of the system while $\mathbf{E}$ is the electric field operator of the environment. The interaction between the system and the reservoir gives rise to irreversible processes such as spontaneous emission, absorption, and dephasing. 

For example, consider a two-level system with ground states $\ket{0}$ and excited state $\ket{1}$. The Pauli operators can be written as
\begin{align}
    \sigma_1=\ket{1}\bra{0}+\ket{1}\bra{0},&\quad \sigma_2=-i\ket{1}\bra{0}+i\ket{1}\bra{0},\nonumber\\
    \quad\sigma_3=&\ket{1}\bra{1}-\ket{0}\bra{0}
\end{align}.
The system Hamiltonian is thus given as $H_S=\frac{1}{2}\omega_{10}\sigma_3$, where $\omega_{10}$ is the transition frequency for the given two-level system. Neglecting the free evolution part, the optical master equation can be written as
\begin{align}
\frac{d\rho(t)}{dt}=\gamma_0 (N+1)\left(\sigma_- \rho(t)\sigma_+ -\frac{1}{2}\left\{\sigma_+\sigma_-,\rho(t)\right\}\right)\nonumber\\
+\gamma_0 N\left(\sigma_+ \rho(t)\sigma_- -\frac{1}{2}\left\{\sigma_-\sigma_+,\rho(t)\right\}\right).\label{Eq:Opt_Mas}
\end{align}

Here, $N=\left(\exp\left(\frac{\hbar\omega_{10}}{k_B T}\right)-1\right)^{-1}$ is the Planck distribution, where \(k_B\) denotes the Boltzmann constant, \(T\) is the temperature of the thermal bath and
\begin{align}
    \sigma_+=\frac{1}{2}(\sigma_1+i\sigma_2),\qquad\sigma_-=\frac{1}{2}(\sigma_1-i\sigma_2)
\end{align}

The parameter \(\gamma_0\) is the spontaneous emission rate, and
\begin{equation}
\gamma=\gamma_0(2N+1)
\end{equation}
is the total emission rate, incorporating both thermally induced emission and absorption processes. For more detail, we refer the reader to \cite{breuer_2002_theory}. 

\section{Main Results}\label{Sec:Main}

In this section, we present the main results of this work. 
\subsection{Quantifying informational completeness quantum measurements}\label{Sec:Info_comp}
Here, we investigate the quantification of the informational completeness of a quantum measurement. We begin with the following definition:
\begin{definition}
Given a quantum measurement 
$A=\{A(x)\}_{x=1}^n\in\mathscr{M}(\cH)$, we define its informational completeness as
\begin{align}
    \mathfrak{D}(A):=\inf_{\substack{\rho,\rho^{\prime}\in\cS(\cH) \\\rho\neq\rho^{\prime}}}\frac{\sum_{x\in\Omega_A}\Big|\tr[\rho A(x)]-\tr[\rho^{\prime} A(x)]\Big|}{||\rho-\rho^{\prime}||_1},\label{Eq:Def_IC_Power}
\end{align}
where $||.||_1$ denotes the trace norm.
\end{definition}
In Eq. \eqref{Eq:Def_IC_Power}, infimum is used  instead of minimum as $\rho^{\prime}$ can vary in the set $\cS(\cH)\setminus\rho$ which is not close and therefore, cannot be compact where the symbol ``$\setminus$" denotes the different between two sets. Now, we begin with proving the unitary invariance of the above measure.
\begin{lemma}
   For an arbitrary measurement $A\in\mathscr{M}(\cH)$, its informational completeness $\mathfrak{D}(A)$ is unitary invariant, i.e.,
\begin{align}
    \mathfrak{D}(U^{\dagger}AU)=\mathfrak{D}(A),
\end{align}
for any unitary matrix $U$.\label{Lem:IC_Pow_Unitary_Inv}
\end{lemma}
\begin{proof}
    From Eq. \eqref{Eq:Def_IC_Power}, we can write
    \begin{align}
         \mathfrak{D}(U^{\dagger}AU)=&\inf_{\substack{\rho,\rho^{\prime}\in\cS(\cH) \\\rho\neq\rho^{\prime}}}\frac{\sum_{x\in\Omega_A}\Big|\tr[\rho U^{\dagger}A(x)U]-\tr[\rho^{\prime} U^{\dagger}A(x)U]\Big|}{||\rho-\rho^{\prime}||_1}\nonumber\\
         =&\inf_{\substack{\rho,\rho^{\prime}\in\cS(\cH) \\\rho\neq\rho^{\prime}}}\frac{\sum_{x\in\Omega_A}\Big|\tr[U\rho U^{\dagger}A(x)]-\tr[U\rho^{\prime} U^{\dagger}A(x)]\Big|}{||\rho-\rho^{\prime}||_1}\nonumber\\
         =&\inf_{\substack{\rho,\rho^{\prime}\in\cS(\cH) \\\rho\neq\rho^{\prime}}}\frac{\sum_{x\in\Omega_A}\Big|\tr[U\rho U^{\dagger}A(x)]-\tr[U\rho^{\prime} U^{\dagger}A(x)]\Big|}{||U(\rho-\rho^{\prime})U^{\dagger}||_1}\nonumber
    \end{align} 
    \begin{align}
         =&\inf_{\substack{\rho,\rho^{\prime}\in\cS(\cH) \\\rho\neq\rho^{\prime}}}\frac{\sum_{x\in\Omega_A}\Big|\tr[U\rho U^{\dagger}A(x)]-\tr[U\rho^{\prime} U^{\dagger}A(x)]\Big|}{||U\rho U^{\dagger}-U\rho^{\prime}U^{\dagger}||_1},
    \end{align}
    where in the second line we have used the cyclicity of the trace, in the third line we have used the fact that the trace norm of an operator is invariant under unitary conjugation. Now, an arbitrary $\rho\in\cS(\cH)$, \emph{if and only if} $\rho_U:=U\rho U^{\dagger}\in\cS(\cH)$ and $\rho-\rho^{\prime}\neq0\iff  \rho_U-\rho^{\prime}_U\neq0$  for arbitrary $\rho,\rho^{\prime}\in\cS(\cH)$, and for an arbitrary unitary matrix $U$, as unitaries are invertible. Thus, we can write
    \begin{align}
        \mathfrak{D}(U^{\dagger}AU)=&\inf_{\substack{\rho_U,\rho_U^{\prime}\in\cS(\cH) \\\rho_U\neq\rho_U^{\prime}}}\frac{\sum_{x\in\Omega_A}\Big|\tr[\rho_U A(x)]-\tr[\rho_U^{\prime} A(x)]\Big|}{||\rho_U-\rho_U^{\prime}||_1}\nonumber\\
        =&\mathfrak{D}(A).
    \end{align}
\end{proof}
Now, we prove the faithfulness of this measure in the following proposition:
\begin{proposition}
     For a generic quantum measurement $B=\{B(x)\}\in\mathscr{M}(\cH)$, the informational completeness $\mathfrak{D}(B)\geq 0$ and $\mathfrak{D}(B)= 0$ if and only if $B$ is informationally incomplete. \label{Prop:IC_meas_positive}
\end{proposition}
\begin{proof}
    From Eq. \eqref{Eq:Def_IC_Power}, it is clear that both numerator and denominator are non-negative. Hence, $\mathfrak{D}(B)\geq 0$ for an arbitrary measurement $B$. Now, $B=\{B(x)\}$ is informationally incomplete then by definition there exists a pair of state $\rho,\rho^{\prime}\in\cS(\cH)$ such that $\tr[\rho A(x)]=\tr[\rho^{\prime} A(x)]~\forall x]\in \Omega_A$. Hence, the numerator $\sum_x|\tr[(\rho-\rho^{\prime})A(x)]|=0$. But as $\rho\neq\rho^{\prime}$ denominator cannot be zero. Hence, $\mathfrak{D}(B)=0$. 

    Now, we have to prove that if $\mathfrak{D}(B)=0$ then $B$ is informationally incomplete.  At first, we observe that if $X\in\cL(\cH)$ is a traceless Hermitian operator, i.e., $X^{\dagger}=X$ and $\tr[X]=0$ such that $||X||_1=1$ then there exists a pair of states $\rho,\rho^{\prime}\in\cS(\cH)$ such that $X=\frac{\rho-\rho^{\prime}}{||\rho-\rho^{\prime}||_1}$ where $\rho\neq \rho^{\prime}$. We use this \emph{observation} in the following steps:

    \begin{align}
        \mathfrak{D}(B)=&\inf_{\substack{\rho,\rho^{\prime}\in\cS(\cH) \\\rho\neq\rho^{\prime}}}\frac{\sum_x\Big|\tr[\rho B(x)]-\tr[\rho^{\prime} B(x)]\Big|}{||\rho-\rho^{\prime}||_1}\nonumber\\
        =&\inf_{\substack{\rho,\rho^{\prime}\in\cS(\cH) \\\rho\neq\rho^{\prime}}}\frac{\sum_x\Big|\tr[Y(\rho,\rho^{\prime}) B(x)]\Big|}{||Y(\rho,\rho^{\prime})||_1}\nonumber\\
        =&\inf_{\substack{\rho,\rho^{\prime}\in\cS(\cH) \\\rho\neq\rho^{\prime}}}\sum_x\Big|\tr[X(\rho,\rho^{\prime}) B(x)]\Big|\nonumber\\
        =&\min_{\substack{X\in\cL(\cH), X^{\dagger}=X \\\tr[X]=0,||X||_1=1}}\sum_x\Big|\tr[X B(x)]\Big|\label{Eq:IC_Power_min_rep}
    \end{align}
    where in the second equality $Y(\rho,\rho^{\prime})=\rho-\rho^{\prime}$, in third equality $X(\rho,\rho^{\prime})=\frac{Y(\rho,\rho^{\prime})}{||Y(\rho,\rho^{\prime})||_1}$ (therefore, $X(\rho,\rho^{\prime})^{\dagger}=X(\rho,\rho^{\prime})$, $\tr[X(\rho,\rho^{\prime})]=0$, and  $||X(\rho,\rho^{\prime})||_1=1$) and in fourth equality we have used above-mentioned observation and the fact that the set $\mathbbm{X}(\cH):=\{X\in\cL(\cH)| X^{\dagger}=X,\tr[X]=0,||X||_1=1\}$ is compact (therefore, one can replace infimum with minimum). Note that $||X||_1=1\Rightarrow X\neq 0 ~\forall X\in\mathbbm{X}(\cH)$. Now, if $\mathfrak{D}(B)=0$ then there exists an $X^*\in\mathbbm{X}(\cH)$ such that $\sum_x\Big|\tr[X^* B(x)]\Big|=0$ which implies $\tr[X^* B(x)]=0\forall~ x\in \Omega_B$ which implies that the set $\{B(x)\}$ cannot span $\cL(\cH)$ as $X\neq 0$. Hence, $B=\{B(x)\}$ is informationally incomplete.



\end{proof}

From Proposition \ref{Prop:IC_meas_positive}, we conclude that the measure $\mathfrak{D}$ is faithful and non-negative. Next, let us focus our attention on the qubit SIC measurements.
\begin{theorem}
    If a qubit measurement $A\in\mathscr{M}(\cH_Q)$ is symmetric informationally complete (SIC), then its informational completeness is
    \begin{align}
        \mathfrak{D}(A)=\frac{1}{\sqrt{6}}.
    \end{align}\label{Th:SIC_IC_Pow}
\end{theorem}
\begin{proof}
Let $\textbf{s}$ and $\textbf{s}^{\prime}$ denote the Bloch vectors corresponding to the states $\rho$ and $\rho^{\prime}$ in Eq. \eqref{Eq:Def_IC_Power}, respectively. Then
\begin{align}
    \rho-\rho^{\prime}
    =&\frac{1}{2}(\textbf{s}-\textbf{s}^{\prime}).\boldsymbol{\sigma}=\frac{\textbf{d}.\boldsymbol{\sigma}}{2},\label{Eq:Den_Diff}
\end{align}
where $\textbf{d}=\textbf{s}-\textbf{s}^{\prime}\in\mathbbm{R}^3$. Consequently,
\begin{align}
    ||\rho-\rho^{\prime}||_1
    =&\tr\Big[\sqrt{(\rho-\rho^{\prime})^{\dagger}(\rho-\rho^{\prime})}\Big]\nonumber\\
    =&\frac{\tr[\sqrt{(\textbf{d}.\boldsymbol{\sigma})(\textbf{d}.\boldsymbol{\sigma})}]}{2}\nonumber\\
    =&\frac{\tr[|\textbf{d}|\Id_{Q}]}{2}=|\textbf{d}|.\label{Eq:IC_power_def_denom}
\end{align}

Next, using Eqs. \eqref{Eq:SIC_Bloch_rep} and \eqref{Eq:Den_Diff}, we obtain
\begin{align}
    \tr[(\rho-\rho^{\prime})A(x)]
    =&\frac{\tr[\sqrt{(\textbf{d}.\boldsymbol{\sigma})(\mathbbm{1}+\textbf{r}(x).\boldsymbol{\sigma})}]}{8},\quad \forall~x\nonumber\\
    =&\frac{\tr[\sqrt{(\textbf{d}.\boldsymbol{\sigma})(\textbf{r}(x).\boldsymbol{\sigma})}]}{8}\nonumber\\
    =&\frac{\textbf{d}.\textbf{r}(x)}{4},\quad \forall~x.\label{Eq:IC_power_def_num}
\end{align}

Substituting Eqs. \eqref{Eq:IC_power_def_denom} and \eqref{Eq:IC_power_def_num} into Eq. \eqref{Eq:Def_IC_Power}, we find
\begin{align}
    \mathfrak{D}(A)
    =&\inf_{\textbf{d}\neq0}\frac{1}{4}\frac{\sum_x\Big|\textbf{d}.\textbf{r}(x)\Big|}{|\textbf{d}|}\nonumber\\
    =&\min_{\hat{|\textbf{d}|}=1}\frac{1}{4}\sum_x\Big|\hat{\textbf{d}}.\textbf{r}(x)\Big|\label{Eq:Def_IC_Power1}
\end{align}

where $\hat{\textbf{d}}$ is the unit vector along $\textbf{d}$. For qubits, all SIC measurements are unitarily equivalent.

Therefore, from Lemma \ref{Lem:IC_Pow_Unitary_Inv}, without loss of generality, we may choose the SIC measurement as in Eqs. \eqref{Eq:SIC_vec12} and \eqref{Eq:Sic_vec34}.

Since $\sum_x \textbf{r}(x)=0$, we have
\begin{align}
    \sum_x\hat{\textbf{d}}.\textbf{r}(x)=0.
\end{align}
Separating the above sum into positive and negative contributions gives
\begin{align}
    \sum_{\substack{x \\\hat{\textbf{d}}.\textbf{r}(x)\geq0}}\hat{\textbf{d}}.\textbf{r}(x)
    +
    \sum_{\substack{x \\\hat{\textbf{d}}.\textbf{r}(x)<0}}\hat{\textbf{d}}.\textbf{r}(x)
    =0.\label{Eq:split_sum}
\end{align}

Define $\sum_{\substack{x \\\hat{\textbf{d}} \textbf{r}(x)\geq0}}\hat{\textbf{d}}.\textbf{r}(x):=p\geq 0$. Then Eq. \eqref{Eq:split_sum} implies $\sum_{\substack{x \\\hat{\textbf{d}}.\textbf{r}(x)<0}}\hat{\textbf{d}}.\textbf{r}(x)=-p$, and hence
\begin{align}
    \sum_x\Big|\hat{\textbf{d}}.\textbf{r}(x)\Big|=2p.\label{Eq:mod_sum_bound}
\end{align}

Now consider
\begin{align}
\sum_x\Big(\hat{\textbf{d}}.\textbf{r}(x)\Big)^2
=&\sum_{i,j}\sum_x\hat{\textbf{d}}_i\hat{\textbf{d}}_j\textbf{r}_i(x)\textbf{r}_j(x)\nonumber\\
=&\sum_x\Big(\sum_i \hat{\textbf{d}}_i^2(\textbf{r}_i(x))^2
+2\sum_{\substack{i,j\\i<j}}
\hat{\textbf{d}}_i\hat{\textbf{d}}_j
\textbf{r}_i(x)\textbf{r}_j(x)\Big).\label{Eq:square_sum_exp}
\end{align}

As $|\hat{\textbf{d}}|=1$, using Eqs. \eqref{Eq:SIC_vec12} and \eqref{Eq:Sic_vec34} in Eq. \eqref{Eq:square_sum_exp}, we obtain
\begin{align}
 \sum_x\Big(\hat{\textbf{d}}.\textbf{r}(x)\Big)^2
 =&\sum_x\frac{|\hat{\textbf{d}}|^2}{3}=\frac{4}{3}.\label{Eq:squared_sum}
\end{align}

For the positive contributions we have
\begin{align}
    \sum_{\substack{x \\\hat{\textbf{d}}.\textbf{r}(x)\geq0}}
    \Big(\hat{\textbf{d}}.\textbf{r}(x)\Big)^2
    \leq&\,
    \Big(
    \sum_{\substack{x \\\hat{\textbf{d}}.\textbf{r}(x)\geq0}}
    \hat{\textbf{d}}.\textbf{r}(x)
    \Big)^2=\,p^2.\label{Eq:squared_sum_exp_pos}
\end{align}

Similarly,
\begin{align}
     \sum_{\substack{x \\\hat{\textbf{d}}.\textbf{r}(x)<0}}
     \Big(\hat{\textbf{d}}.\textbf{r}(x)\Big)^2
     \leq&
     \Big(
     \sum_{\substack{x \\\hat{\textbf{d}}.\textbf{r}(x)<0}}
     \hat{\textbf{d}}.\textbf{r}(x)
     \Big)^2=p^2.\label{Eq:squared_sum_exp_neg}
\end{align}

Separating the sum in L.H.S. in Eq. \eqref{Eq:squared_sum} into posiitve and negative contributions and substituting Eqs. \eqref{Eq:squared_sum_exp_pos} and \eqref{Eq:squared_sum_exp_neg} in it we conclude that
\begin{align}
    p^2\geq\frac{2}{3}.
\end{align}
Using Eq. \eqref{Eq:mod_sum_bound}, we obtain the bound
\begin{align}
    \sum_x\Big|\hat{\textbf{d}}.\textbf{r}(x)\Big|
    =2p
    \geq\frac{2\sqrt{2}}{\sqrt{3}},
    \qquad
    \forall~\hat{\textbf{d}},~|\hat{\textbf{d}}|=1.
\end{align}

Consequently, Eq. \eqref{Eq:Def_IC_Power1} yields
\begin{align}
    \mathfrak{D}(A)
    =\min_{\hat{|\textbf{d}|}=1}\frac{1}{4}\sum_x\Big|\hat{\textbf{d}}.\textbf{r}(x)\Big|
    \geq\frac{1}{\sqrt{6}}.\label{Eq:IC_power_bound}
\end{align}

We now show that this bound is achievable. Choosing
\begin{align}
    \hat{\textbf{d}}=\frac{1}{\sqrt{2}}(1,-1,0),
\end{align}
and using Eqs. \eqref{Eq:SIC_vec12} and \eqref{Eq:Sic_vec34}, we obtain
\begin{align}
    \frac{1}{4}\sum_x\Big|\hat{\textbf{d}}.\textbf{r}(x)\Big|
    =&\frac{1}{4\sqrt{6}}
    (|1-1|+|1+1|+|-1-1|+|-1+1|)\nonumber\\
    =&\frac{1}{\sqrt{6}}.
\end{align}

Thus, Eq. \eqref{Eq:IC_power_bound} is saturated, and hence the IC power of a qubit SIC measurement is
\begin{align}
    \mathfrak{D}(A)=\frac{1}{\sqrt{6}}.
\end{align}
\end{proof}

With the informational completeness of a qubit SIC measurement established, we now turn to the case of general minimal informationally complete qubit measurements. In particular, we prove the following theorem.

\begin{theorem}
    Let $A\in\mathscr{M}(\cH_Q)$ be any four-outcome qubit measurement. Then its informational completeness satisfies
    \begin{align}
        \mathfrak{D}(A)\leq\frac{1}{\sqrt{6}},
    \end{align}
    where inequality is saturated when $A$ is an SIC qubit measurement. \label{Th:minimal_IC_SIC}
\end{theorem}

\begin{proof}
Consider a general four outcome qubit measurement qubit measurement $A=\{A(x)\}_{x=1}^4$, which can be written as
\begin{align}
    A(x)=\frac{1}{2}\Big(a(x)+\textbf{b}(x).\boldsymbol{\sigma}\Big),
    ~\forall x.\label{Eq:minIC_Bloch_rep}
\end{align}

As $A$ is a valid measurement, the parameters must satisfy
\begin{align}
    \sum_x a(x)=2, ~\sum_x\textbf{b}(x)=0, ~|\textbf{b}(x)|\leq a(x), ~\forall ~x.\label{Eq:bloch_vec_prop}
\end{align}

Note that if $A$ is informationally incomplete $\mathfrak{D}(A)=0$ and hence, the theorem is trivially satisfied. Therfore, we assume $A$ to be a four outcome (i.e., minimal) informationally complete qubit measurement. Now, substituting Eq. \eqref{Eq:minIC_Bloch_rep} into Eq. \eqref{Eq:Def_IC_Power}, and following steps analogous to the SIC case, we obtain
\begin{align}
    \mathfrak{D}(A)
    =
    \min_{\hat{|\textbf{d}|}=1}
    \frac{1}{2}
    \sum_x\Big|\hat{\textbf{d}}.\textbf{b}(x)\Big|=:\frac{\Delta}{2}.\label{Eq:IC_power_minIC}
\end{align}

Clearly, here, $\Delta=\min_{\hat{|\textbf{d}|}=1}\sum_x\Big|\hat{\textbf{d}}.\textbf{b}(x)\Big|$. Now choose $\hat{\textbf{d}}$ such that
\begin{align}
   \hat{\textbf{d}} \perp \textbf{b}(i),\textbf{b}(j),
\end{align}
for some $i,j\in\{1,2,3,4\}$. Then $\hat{\textbf{d}}.\textbf{b}(i)=0$, and $\hat{\textbf{d}}.\textbf{b}(j)=0$.

Using $\sum_x\textbf{b}(x)=0$, we obtain
\begin{align}
    0
    =&\sum_x \hat{\textbf{d}}.\textbf{b}(x)\nonumber\\
    =&\hat{\textbf{d}}.\textbf{b}(k)+\hat{\textbf{d}}.\textbf{b}(l)\nonumber\\
    \Longrightarrow&
    \hat{\textbf{d}}.\textbf{b}(k)
    =
    -\hat{\textbf{d}}.\textbf{b}(l),
\end{align}
where $k,l\in\{1,2,3,4\}\setminus\{i,j\}$.

Hence,
\begin{align}
    \sum_x \Big|\hat{\textbf{d}}.\textbf{b}(x)\Big|
    =2\Big|\hat{\textbf{d}}.\textbf{b}(k)\Big|.
\end{align}

Therefore,
\begin{align}
    \Delta
    \leq&
    \sum_x \Big|\hat{\textbf{d}}.\textbf{b}(x)\Big|=2\Big|\hat{\textbf{d}}.\textbf{b}(k)\Big|=\Big|\hat{\textbf{d}}.(\textbf{b}(k)-\textbf{b}(l))\Big|\nonumber\\
    \leq&|\hat{\textbf{d}}||\textbf{b}(k)-\textbf{b}(l)|=|\textbf{b}(k)-\textbf{b}(l)|,
\end{align}
where the second inequality follows from the Cauchy--Schwarz inequality and in the last equality, we have used $|\hat{\textbf{d}}|=1$. Since this holds for arbitrary $(k,l)$,
\begin{align}
    \Delta\leq\min_{k<l}|\textbf{b}(k)-\textbf{b}(l)|.\label{Eq:Delta_diff_rel}
\end{align}

Next, observe that as for $k<l$, six $(k,l)$ pairs are possible, we have
\begin{align}
    \min_{k<l}|\textbf{b}(k)-\textbf{b}(l)|^2
    \leq
    \frac{1}{6}
    \sum_{\substack{k,l\\k<l}}
    |\textbf{b}(k)-\textbf{b}(l)|^2.\label{Eq:min_comp_diff}
\end{align}

Expanding the right-hand side gives
\begin{align}
    \sum_{\substack{k,l\\k<l}}
    |\textbf{b}(k)-\textbf{b}(l)|^2
    =&
    \sum_{\substack{k,l\\k<l}}
    (|\textbf{b}(k)|^2+|\textbf{b}(l)|^2-2\textbf{b}(k).\textbf{b}(l))\nonumber\\
    =&3\sum_{x=1}^4|\textbf{b}(x)|^2
    -2\sum_{\substack{k,l\\k<l}}
    \textbf{b}(k).\textbf{b}(l).\label{Eq:diff_exp}
\end{align}

Moreover,
\begin{align}
    0
    =
    \Big|\sum_x\textbf{b}(x)\Big|^2
    =&
    \sum_x|\textbf{b}(x)|^2
    +2\sum_{\substack{k,l\\k<l}}
    \textbf{b}(k).\textbf{b}(l)\nonumber\\
    \Longrightarrow&
    \sum_{\substack{k,l\\k<l}}
    \textbf{b}(k).\textbf{b}(l)
    =
    -\frac{\sum_x|\textbf{b}(x)|^2}{2}.\label{Eq:prod_def}
\end{align}

Substituting Eq. \eqref{Eq:prod_def} into Eq. \eqref{Eq:diff_exp}, and then using Eq. \eqref{Eq:bloch_vec_prop}, we obtain
\begin{align}
    \sum_{\substack{k,l\\k<l}}
    |\textbf{b}(k)-\textbf{b}(l)|^2
    =&4\sum_{x=1}^4|\textbf{b}(x)|^2\nonumber\\
    \leq&4\sum_{x=1}^4(a(x))^2.\label{Eq:diff_ineq}
\end{align}

Combining Eqs. \eqref{Eq:IC_power_minIC}, \eqref{Eq:Delta_diff_rel}, \eqref{Eq:min_comp_diff}, and \eqref{Eq:diff_ineq}, we obtain
\begin{align}
    \mathfrak{D}(A)
    =
    \frac{\Delta}{2}
    \leq
    \frac{1}{\sqrt{6}}
    \sqrt{\sum_{x=1}^4(a(x))^2}.\label{Eq:minIC_power_bound}
\end{align}

We now determine when the inequalities are saturated. From Eq. \eqref{Eq:bloch_vec_prop}, we observe that the inequality in Eq. \eqref{Eq:diff_ineq} is saturated only if
\begin{align}
    |\textbf{b}(x)|=a(x),
    \qquad \forall ~x.\label{Eq:first_inq_sat}
\end{align}

The inequality in Eq. \eqref{Eq:min_comp_diff} is saturated only when
\begin{align}
    &|\textbf{b}(k)-\textbf{b}(l)|
    :=d,
    \qquad \forall~k\neq l\nonumber\\
\Rightarrow&|\textbf{b}(k)|^2+|\textbf{b}(l)|^2
    -2\textbf{b}(k).\textbf{b}(l)
    =d^2,
    \qquad \forall~k\neq l.\label{Eq:diff_exp_eq}
\end{align}

Since $\sum_x\textbf{b}(x)=0$, for arbitrary $k$ we have
\begin{align}
    \textbf{b}(k).\sum_x\textbf{b}(x)
    =&0,
    \qquad\forall~k\nonumber\\
    \Big|\textbf{b}(k)\Big|^2
    +
    \sum_{\substack{l\\l\neq k}}
    \textbf{b}(k).\textbf{b}(l)
    =&0,
    \qquad\forall~k.
\end{align}

Using Eq. \eqref{Eq:diff_exp_eq}, this becomes
\begin{align}
     \Big|\textbf{b}(k)\Big|^2
     +
     \sum_{\substack{l\\l\neq k}}
     \frac{
     |\textbf{b}(k)|^2+|\textbf{b}(l)|^2-d^2
     }{2}
     =0,
     \qquad\forall~k\nonumber\\
     |\textbf{b}(k)|^2
     =
     \frac{
     3d^2-\sum_x|\textbf{b}(x)|^2
     }{4},
     \qquad\forall~k.\label{Eq:Bloch_vec_eq}
\end{align}

From Eq. \eqref{Eq:Bloch_vec_eq}, we conclude that $|\textbf{b}(k)|=|\textbf{b}(l)|~\forall k,l$. Therefore, Eqs. \eqref{Eq:bloch_vec_prop} and \eqref{Eq:first_inq_sat}, imply
\begin{align}
    |\textbf{b}(x)|=a(x)=\frac{1}{2},
    \qquad\forall~x.\label{Eq:Bloch_vec_b_eq_a}
\end{align}

Again using Eq. \eqref{Eq:Bloch_vec_eq}, we obtain
\begin{align}
    \frac{1}{4}
    =&\frac{3d^2-1}{4}\nonumber\\
    d=&\sqrt{\frac{2}{3}}.
\end{align}

Substituting this into Eq. \eqref{Eq:diff_exp_eq} yields
\begin{align}
    \textbf{b}(k).\textbf{b}(l)
    =&-\frac{1}{12},
    \qquad \forall~k\neq l.
\end{align}

Since $|\textbf{b}(x)|=a(x)$ for every $x$, we may write
\begin{align}
    \textbf{b}(x)=a(x)\hat{\textbf{b}}(x),
    \qquad\forall ~x.\label{Eq:bloch_vec_scal}
\end{align}

Therefore,
\begin{align}
    -\frac{1}{12}
    =&\textbf{b}(k).\textbf{b}(l),
    \qquad \forall~k\neq l\nonumber\\
    =&a(k)a(l)\hat{\textbf{b}}(k).\hat{\textbf{b}}(l),
    \qquad \forall~k\neq l\nonumber\\
    =&\frac{1}{4}\hat{\textbf{b}}(k).\hat{\textbf{b}}(l),
    \qquad \forall~k\neq l\nonumber\\
    \Longrightarrow&
    \hat{\textbf{b}}(k).\hat{\textbf{b}}(l)
    =
    -\frac{1}{3},
    \qquad \forall~k\neq l.\label{Eq:SIC_inner_prod}
\end{align}

Finally, from Eqs. \eqref{Eq:minIC_power_bound} and \eqref{Eq:Bloch_vec_b_eq_a} we conclude that
\begin{align}
    \mathfrak{D}(A)=\frac{\Delta}{2}\leq\frac{1}{\sqrt{6}}.\label{Eq:minIC_up_bound}
\end{align}

Furthermore, substituting Eqs. \eqref{Eq:Bloch_vec_b_eq_a}, \eqref{Eq:bloch_vec_scal} in Eq. \eqref{Eq:minIC_Bloch_rep} we get that the measurement that can achieve this bound is of the form
\begin{align}
     A(x)=\frac{1}{4}\Big(\mathbbm{1}+\hat{\textbf{b}}(x).\boldsymbol{\sigma}\Big),
    ~\forall x,
\end{align}
which, along with Eq. \eqref{Eq:SIC_inner_prod}, is an SIC qubit measurement (see Eq. \eqref{Eq:SIC_Bloch_rep}).

Theorem \ref{Th:SIC_IC_Pow} indeed confirms this that for qubit SIC measurements $\mathfrak{D}(A)=\frac{1}{\sqrt{6}}$, and hence, the bound in Eq. \eqref{Eq:minIC_up_bound} is tight.
\end{proof}

Theorem \ref{Th:minimal_IC_SIC} is consistent with the known fact that when the given state is completely unknown, for quantum state tomography, SICs are the best among minimal ICs.

\subsection{Informational completeness preservability of Quantum Channels}\label{Sec:Info_Comp_Pres}
In this section, we explore the quantification of the informational completeness preservability of a quantum channel. We begin with the following definition:
\begin{definition}
    Given a quantum channel $\Lambda\in\mathscr{C}(\cH,\cH)$, we define its informational completeness preservability (or IC- preservability) as
    \begin{align}
        \tilde{\mathfrak{D}}(\Lambda):=
        \sup_{\substack{A=\{A(x)\}\\
        A\in\mathscr{M}(\cH)}}
        \inf_{\substack{\rho,\rho^{\prime}\in\cS(\cH)\\\rho\neq\rho^{\prime}}}
        \frac{
        \sum_x\Big|
        \tr[\rho \Lambda^{\dagger}(A(x))]
        -
        \tr[\rho^{\prime} \Lambda^{\dagger}(A(x))]
        \Big|
        }{
        ||\rho-\rho^{\prime}||_1
        }.
        \label{Eq:Def_IC_Pres_Power}
    \end{align}
\end{definition}
Note that Eq. \eqref{Eq:Def_IC_Pres_Power} can equivalently be written as 
\begin{equation}
    \tilde{\mathfrak{D}}(\Lambda)=\sup_{\substack{A=\{A(x)\}, A\in\mathscr{M}(\cH)}}\mathfrak{D}(\Lambda^{\dagger}(A)).\label{Eq:Def_IC_Pres_Power_simplified}
\end{equation}
\begin{remark}
\rm{From Eq. \eqref{Eq:Def_IC_Pres_Power_simplified}, using Lemma \ref{Lem:IC_Pow_Unitary_Inv} and the fact that for any given measurement $A=\{A(x)\}\in\mathscr{M}(\cH)$ and a unitary conjugation channel $\mathcal{V}$, $\mathcal{V}(A)=\{VA(x)V^{\dagger}\}$, where $V$ is a unitary matrix, is also a valid measurement, it immediately follows that the IC-preservibility of a quantum channel is invariant under unitary pre and post-processing. That is,
\begin{align}
    \tilde{\mathfrak{D}}(\Lambda)
    =
    \tilde{\mathfrak{D}}(\mathcal{U}\circ\Lambda\circ\mathcal{V}),
\end{align}
where $\mathcal{U}$ and $\mathcal{V}$ denote unitary conjugation channels.}\label{Rem:IC-Pres_Pow_Unitary_Inv}    
\end{remark}

\begin{proposition}
     For a generic quantum channel $\Lambda\in\mathscr{C}(\cH,\cH)$, its IC-preservability $\tilde{\mathfrak{D}}(\Lambda)\geq 0$ and $\tilde{\mathfrak{D}}(\Lambda)= 0$ if and only if $\Lambda^{\dagger}(A(x)$ is informationally incomplete for all $A=\{A(x)\}\in\mathscr{M}(\cH)$. \label{Prop:IC_pres_pow_positive}
\end{proposition}

\begin{proof}
    As $\Lambda$ is a quantum channel, $\Lambda^{\dagger}(A)$ is a valid measurement for any valid measurement $A\in\mathscr{M}(\cH)$. From Eq. \eqref{Eq:Def_IC_Pres_Power_simplified} and Proposition \ref{Prop:IC_meas_positive} we observe that
    \begin{align}
        \hat{\mathfrak{D}}(\Lambda)=&\sup_{\substack{A=\{A(x)\}, A\in\mathscr{M}(\cH)}}\mathfrak{D}(\Lambda^{\dagger}(A))\nonumber\\
        \geq&\mathfrak{D}(\Lambda^{\dagger}(B))\nonumber\\
        \geq&0,
    \end{align}
where $B\in\mathscr{M}(\cH)$ is an arbitrary measurement. Next, if $\hat{\mathfrak{D}}(\Lambda)=0$ then as $\hat{\mathfrak{D}}(\Lambda)\geq0$, we get $\mathfrak{D}(\Lambda^{\dagger}(A))=0$ for any $A\in\mathscr{M}(\cH)$, From Proposition \ref{Prop:IC_meas_positive} we can conclude that $\Lambda^{\dagger}(A)$ is informationally incomplete for any $A\in\mathcal{M}(\cH)$. Conversely, if $\Lambda^{\dagger}(A)$ is informationally incomplete for any $A\in\mathcal{M}(\cH)$ then, again from Proposition \ref{Prop:IC_meas_positive} we get $\mathfrak{D}(\Lambda^{\dagger}(A))=0$ for any $A\in\mathcal{M}(\cH)$. Hence, from Eq. \eqref{Eq:Def_IC_Pres_Power_simplified}, we have $\hat{\mathfrak{D}}(\Lambda)=0$.
\end{proof}

Proposition \ref{Prop:IC_pres_pow_positive}, establishes the faithfulness and non-negativity of $\tilde{\mathfrak{D}}(\Lambda)$. 

Next, we establish the following monotonicity properties of the IC-preservability.

\begin{proposition}
Let $\Lambda_1\in\mathscr{C}(\cH,\cH)$ and $\Lambda_2\in\mathscr{C}(\cH,\cH)$ be two quantum channels.
\begin{enumerate}[label=P \arabic*.]
    \item If $\Lambda_2\succeq_{postproc}\Lambda_1$, then
    \begin{align}
        \tilde{\mathfrak{D}}(\Lambda_2)\geq\tilde{\mathfrak{D}}(\Lambda_1).
    \end{align}\label{Eq:IC_Pres_Pow_Post_Proc}

    \item If there exists a statistical morphism $\mathcal{N}$ such that $\Lambda_1=\mathcal{N}\circ\Lambda_2$,
    then
    \begin{align}
        \tilde{\mathfrak{D}}(\Lambda_2)\geq\tilde{\mathfrak{D}}(\Lambda_1).
    \end{align}\label{Eq:IC_Pres_Pow_Stat_Morph}

    \item If $\Lambda_2\succeq_{asymp}\Lambda_1$, then
    \begin{align}
        \tilde{\mathfrak{D}}(\Lambda_2)=0\Longrightarrow\tilde{\mathfrak{D}}(\Lambda_1)=0.
    \end{align}\label{Eq:IC_Pres_Pow_Post_Proc}
\end{enumerate}
\end{proposition}

\begin{proof}
    We start with proving the first statement. Suppose that
    \begin{align}
        \Lambda_2\succeq_{postproc}\Lambda_1.
    \end{align}
    Then, there exists a quantum channel $\Phi$ such that
    \begin{align}
        \Lambda_1=\Phi\circ\Lambda_2.
    \end{align}

    Using Eq. \eqref{Eq:Def_IC_Pres_Power}, we obtain
    \begin{align}
        &\tilde{\mathfrak{D}}(\Lambda_1)\nonumber\\
        =&
        \sup_{\substack{A=\{A(x)\}\\A\in\mathscr{M}(\cH)}}
        \inf_{\substack{\rho,\rho^{\prime}\\\rho\neq\rho^{\prime}}}
        \frac{
        \sum_x\Big|
        \tr[\rho \Lambda_1^{\dagger}(A(x))]
        -
        \tr[\rho^{\prime} \Lambda_1^{\dagger}(A(x))]
        \Big|
        }{
        ||\rho-\rho^{\prime}||_1
        }\nonumber\\
        =&
        \sup_{\substack{A=\{A(x)\}\\A\in\mathscr{M}(\cH)}}
        \inf_{\substack{\rho,\rho^{\prime}\\\rho\neq\rho^{\prime}}}
        \frac{
        \sum_x\Big|
        \tr[\rho \Lambda_2^{\dagger}\circ\Phi^{\dagger}(A(x))]
        -
        \tr[\rho^{\prime} \Lambda_2^{\dagger}\circ\Phi^{\dagger}(A(x))]
        \Big|
        }{
        ||\rho-\rho^{\prime}||_1
        }\nonumber\\
        =&
        \sup_{\substack{A=\{A(x)\}\\A\in\mathscr{M}(\cH)}}
        \inf_{\substack{\rho,\rho^{\prime}\\\rho\neq\rho^{\prime}}}
        \frac{
        \sum_x\Big|
        \tr[\rho \Lambda_2^{\dagger}(\Phi^{\dagger}(A(x)))]
        -
        \tr[\rho^{\prime} \Lambda_2^{\dagger}(\Phi^{\dagger}(A(x)))]
        \Big|
        }{
        ||\rho-\rho^{\prime}||_1
        }\nonumber\\
        \leq&
        \sup_{\substack{B=\{B(x)\}\\B\in\mathscr{M}(\cH)}}
        \inf_{\substack{\rho,\rho^{\prime}\\\rho\neq\rho^{\prime}}}
        \frac{
        \sum_x\Big|
        \tr[\rho \Lambda_2^{\dagger}(B(x))]
        -
        \tr[\rho^{\prime} \Lambda_2^{\dagger}(B(x))]
        \Big|
        }{
        ||\rho-\rho^{\prime}||_1
        }\nonumber\\
        =&\tilde{\mathfrak{D}}(\Lambda_2)
    \end{align}
    where the inequality follows from the fact that $\{\Phi^{\dagger}(A)\}$ is itself a valid measurement and from the fact that the set $\{\Phi^{\dagger}(A)|A\in\mathscr{M}(\cH)\}\subseteq\mathscr{M}(\cH)$. Clearly, the equality is achieved whenever $\Phi$ is a unitary channel.

    We now prove the second statement. Suppose that $\mathcal{N}$ is a statistical morphism satisfying
    \begin{align}
        \Lambda_1=\mathcal{N}\circ\Lambda_2.
    \end{align}
    
    Then, from Eq. \eqref{Eq:Stat_morph} we have,
    \begin{align}
        &\tilde{\mathfrak{D}}(\Lambda_1)\nonumber\\
        =&
        \sup_{\{A(x)\}}
        \inf_{\substack{\rho,\rho^{\prime}\\
        \rho\neq\rho^{\prime}}}
        \frac{
        \sum_x\Big|
        \tr[\rho \Lambda_1^{\dagger}(A(x))]
        -
        \tr[\rho^{\prime} \Lambda_1^{\dagger}(A(x))]
        \Big|
        }{
        ||\rho-\rho^{\prime}||_1
        }\nonumber\\
        =&
        \sup_{\{A(x)\}}
        \inf_{\substack{\rho,\rho^{\prime}\\\rho\neq\rho^{\prime}}}
        \frac{
        \sum_x\Big|
        \tr[\rho \Lambda_2^{\dagger}\circ\mathcal{N}^{\dagger}(A(x))]
        -
        \tr[\rho^{\prime} \Lambda_2^{\dagger}\circ\mathcal{N}^{\dagger}(A(x))]
        \Big|
        }{
        ||\rho-\rho^{\prime}||_1
        }\nonumber\\
        =&
        \sup_{\{A(x)\}}
        \inf_{\substack{\rho,\rho^{\prime}\\\rho\neq\rho^{\prime}}}
        \frac{
        \sum_x\Big|
        \tr[\mathcal{N}\circ\Lambda_2(\rho) A(x)]
        -
        \tr[\mathcal{N}\circ\Lambda_2(\rho^{\prime})A(x)]
        \Big|
        }{
        ||\rho-\rho^{\prime}||_1
        }\nonumber\\
        =&\sup_{\{A(x)\}}
        \inf_{\substack{\rho,\rho^{\prime}\\\rho\neq\rho^{\prime}}}
        \frac{
        \sum_x\Big|
        \tr[\Lambda_2(\rho) A_{\cN}(x)]
        -
        \tr[\Lambda_2(\rho^{\prime})A_{\cN}(x)]
        \Big|
        }{
        ||\rho-\rho^{\prime}||_1
        }\nonumber\\
        \leq&
        \sup_{\{B(x)\}}
        \inf_{\substack{\rho,\rho^{\prime}\\\rho\neq\rho^{\prime}}}
        \frac{
        \sum_x\Big|
        \tr[\Lambda_2(\rho) B(x)]
        -
        \tr[\Lambda_2(\rho^{\prime}) B(x)]
        \Big|
        }{
        ||\rho-\rho^{\prime}||_1
        }\nonumber\\
        =&\tilde{\mathfrak{D}}(\Lambda_2)\label{Eq:Mat_forms}
    \end{align}
    where the inequality follows from the fact that the set $\{A_{\cN}|A\in\mathscr{M}(\cH)\}\subseteq\mathscr{M}(\cH)$.

    Finally, let us proceed with the proof of the third statement.

    Let us assume that $\tilde{\mathfrak{D}}(\Lambda_2)=0$. Then from Eq. \eqref{Eq:Def_IC_Pres_Power} it implies that for any informationally complete measurement $A=\{A(x)\}$ there exist as a pair of states $\rho,\rho^{\prime}\,\text{where}\,\rho\neq\rho^{\prime}$, such that
    \begin{align}
        &\sum_x\Big|
        \tr[\Lambda_2(\rho) A(x)]
        -\tr[\Lambda_2(\rho^{\prime}) A(x)]
        \Big|=0
        \nonumber\\
        \Longrightarrow&\tr[\Lambda_2(\rho-\rho^{\prime}) A(x)]=0\qquad\forall\,x\nonumber\\\Longrightarrow&\Lambda_2(\rho-\rho^{\prime})=0,
    \end{align}
where third line follows form the fact that $A=\{A(x)\}$ is informationally complete. From Eq. \eqref{Eq:IC_Power_min_rep} it is clear that such $\rho$ and $\rho^{\prime}$ always exist as \emph{inf} in Eq. \eqref{Eq:Def_IC_Power} can be replaced with minimum, and hence, is achievable. Thus we can conclude that $||\Lambda_2(\rho-\rho^{\prime})||_1=0$. Now, as $\Lambda_2\succeq_{asymp}\Lambda_1$, from \cite{Mitra_2026} (see Theorem 3) it implies that
\begin{align}
    ||\Lambda_1(\rho-\rho^{\prime})||_1=0.
\end{align}
Considering that fact that $\tilde{\mathfrak{D}}(\Lambda_1)$ is upper-bounded as
\begin{align}
    \tilde{\mathfrak{D}}(\Lambda_1)\leq\inf_{\substack{\rho,\rho^{\prime}\\\rho\neq\rho^{\prime}}}
        \frac{
        ||\Lambda_1(\rho-\rho^{\prime})||_1
        }{
        ||\rho-\rho^{\prime}||_1
        },
\end{align}
and $\tilde{\mathfrak{D}}(\Lambda_1)\geq0$, we can conclude that $\tilde{\mathfrak{D}}(\Lambda_1)=0$.
    
    This completes the proof.
\end{proof}
Next, we investigate bounds on the IC-preservability of quantum channels. In order to do that, we prove two following important lemmas

\begin{lemma}
    For any qubit channel $\Lambda$, the IC-preservability is upper bounded as,
    \begin{align}
        \tilde{\mathfrak{D}}(\Lambda)\leq\min_{i\in\{1,2,3\}}|\lambda_i|=:\lambda_{min},
    \end{align}
    where $\lambda_1,\lambda_2, \lambda_3$ are the signed singular values associated with the given qubut channel $\Lambda$.\label{Lem:IC_Pres_up_Bound}
\end{lemma}
\begin{proof}
    Without loss of generality, we can assume that the signed singular values associated with the given channel $\Lambda$ satisfies the ordering $|\lambda_1|\geq|\lambda_2|\geq|\lambda_3|$. Therefore, $\lambda_{min}=|\lambda_3|$.
    Any qubit quantum channel $\Lambda$ admits a decomposition of the form\cite{braun_2014_universal}
    \begin{align}
        \Lambda=\mathcal{U}_{\Lambda}\circ\Lambda_D\circ\mathcal{V}_{\Lambda},\label{Eq:Channel_decomp}
    \end{align}
    where $\mathcal{U}_{\Lambda}$ and $\mathcal{V}_{\Lambda}$ are unitary conjugate channels and $\Lambda_D$ is the quantum channel with diagonal $M_{\Lambda_D}$. Accordingly, any qubit state $\rho$ with Bloch vector $\textbf{s}$ transforms under $\Lambda$, up to unitary rotations, as
    \begin{align}
        \textbf{s}\mapsto M_{\Lambda_D}(\textbf{s})+\textbf{t}.\label{Eq:Bloch_vec_trans_chan}
    \end{align}
    Here,
    \begin{align}
        M_{\Lambda_D}=
        \begin{pmatrix}
        \lambda_1 & 0 & 0 \\
        0 & \lambda_2 & 0 \\
        0 & 0 & \lambda_3
        \end{pmatrix},
        \qquad
        \textbf{t}=
        \begin{pmatrix}
        t_1 \\
        t_2 \\
        t_3
        \end{pmatrix}.
    \end{align}

    Using the unitary invariance of the IC-preservability, we obtain
    \begin{align}
        \tilde{\mathfrak{D}}(\Lambda)
        &=
        \tilde{\mathfrak{D}}(\Lambda_D)\nonumber\\
        &=
        \sup_{A=\{A(x)\}}\inf_{\substack{\rho,\rho^{\prime}\\\rho\neq\rho^{\prime}}}
        \frac{
        \sum_x\Big|\tr[\rho \Lambda_D^{\dagger}(A(x))]-\tr[\rho^{\prime} \Lambda_D^{\dagger}(A(x))]
        \Big|
        }{
        ||\rho-\rho^{\prime}||_1
        }\nonumber\\
        &=\sup_{A=\{A(x)\}}\inf_{\substack{\rho,\rho^{\prime}\\\rho\neq\rho^{\prime}}}\frac{\sum_x\Big|\tr[\Lambda_D(\rho) (A(x))]-\tr[\Lambda_D(\rho^{\prime})(A(x))]\Big|}{||\rho-\rho^{\prime}||_1}\nonumber\\
        &\leq\sup_{A=\{A(x)\}}\inf_{\substack{\rho,\rho^{\prime}\\\rho\neq\rho^{\prime}}}\frac{||\Lambda_D(\rho)-\Lambda_D(\rho^{\prime})||_1}{||\rho-\rho^{\prime}||_1}\nonumber\\
        &=\inf_{\substack{\rho,\rho^{\prime}\\\rho\neq\rho^{\prime}}}\frac{||\Lambda_D(\rho)-\Lambda_D(\rho^{\prime})||_1}{||\rho-\rho^{\prime}||_1},
    \end{align}
    where in the first inequality we have used the fact that the numerator is always less than the trace distance, and in the last equality we have used the fact that the objective function is independent of $A$.
    
   If $\textbf{s}$ and $\textbf{s}^{\prime}$ are the Bloch vectors corresponding to $\rho$ and $\rho^{\prime}$, respectively, and $\textbf{d}:=\textbf{s}-\textbf{s}^{\prime}$ Eq. \eqref{Eq:IC_power_def_denom} and Eq. \eqref{Eq:Bloch_vec_trans_chan} imply
\begin{align}
    \inf_{\substack{\rho,\rho^{\prime}\\\rho\neq\rho^{\prime}}}
    \frac{||\Lambda_D(\rho)-\Lambda_D(\rho^{\prime})||_1}{||\rho-\rho^{\prime}||_1}&=\inf_{\substack{\textbf{s},\textbf{s}^{\prime}\\ \textbf{s}\neq\textbf{s}^{\prime}}}
    \frac{|M_{\Lambda_D}(\textbf{s}-\textbf{s}^{\prime})|}{|\textbf{s}-\textbf{s}^{\prime}|}\nonumber\\
    &=\inf_{k>0}\min_{|\textbf{d}|=k}
    \frac{|M_{\Lambda_D}(\textbf{d})|}{|\textbf{d}|}\nonumber\\
    &=\inf_{k>0}\min_{|\textbf{d}|=k}
    \frac{\sqrt{\lambda_1^2d_1^2+\lambda_2^2d_2^2+\lambda_3^2d_3^2}}{k},
\end{align}

where $\textbf{d}=(d_1,d_2,d_3)$. Since $|\lambda_1|\geq|\lambda_2|\geq|\lambda_3|$, the minimum value $|\lambda_3|$ is attained for vector of the form $\textbf{d}=(0,0,k)$ and we observe that the minimum value is independent of $|\textbf{d}|=k$. Therefore,
\begin{align}
    \tilde{\mathfrak{D}}(\Lambda)
    &\leq
    \inf_{\substack{\rho,\rho^{\prime}\\\rho\neq\rho^{\prime}}}
    \frac{
    ||\Lambda_D(\rho)-\Lambda_D(\rho^{\prime})||_1
    }{
    ||\rho-\rho^{\prime}||_1
    }\nonumber\\
    &=\inf_{k>0}|\lambda_3|=|\lambda_3|=\lambda_{min}.\label{Eq:IC_Pres_Pow_Up_Bound}
\end{align}
\end{proof}
\begin{lemma}
     For any qubit channel $\Lambda$, the IC-preservability is lower-bounded as
    \begin{align}
        \tilde{\mathfrak{D}}(\Lambda)\geq\frac{\max(0,\lambda_{min}-|\textbf{t}|)}{\sqrt{6}}.
    \end{align}\label{Lem:IC_pres_lower_bound}
\end{lemma}

\begin{proof}
    From Eq. \eqref{Eq:Def_IC_Pres_Power}, it follows immediately that
    \begin{align}
        \tilde{\mathfrak{D}}(\Lambda)
        &=
        \sup_{\{A(x)\}}
        \inf_{\substack{\rho,\rho^{\prime}\\\rho\neq\rho^{\prime}}}
        \frac{
        \sum_x\Big|
        \tr[\rho \Lambda^{\dagger}(A(x))]
        -
        \tr[\rho^{\prime} \Lambda^{\dagger}(A(x))]
        \Big|
        }{
        ||\rho-\rho^{\prime}||_1
        }\nonumber\\
        &\geq
        \inf_{\substack{\rho,\rho^{\prime}\\\rho\neq\rho^{\prime}}}
        \frac{
        \sum_x\Big|
        \tr[\rho \Lambda^{\dagger}(A^{\prime}(x))]
        -
        \tr[\rho^{\prime} \Lambda^{\dagger}(A^{\prime}(x))]
        \Big|
        }{
        ||\rho-\rho^{\prime}||_1
        },
    \end{align}
    where $A^{\prime}=\{A^{\prime}(x)\}$ is a qubit SIC measurement of the form given in Eq. \eqref{Eq:SIC_Bloch_rep}. Let us choose the qubit SIC measurement such that its Bloch vectors are of the form in Eqs. \eqref{Eq:SIC_vec12} and \eqref{Eq:Sic_vec34}. If $\textbf{d}^{\prime}=(M_{\Lambda_D}(\hat{\textbf{d}})+\textbf{t})$, $\{r(x)\}$ are the Bloch vectors of qubit SIC measurement and $\hat{\textbf{d}}$ is an arbitrary unit vector on the Bloch sphere then from Eq. \eqref{Eq:square_sum_exp} we have
     \begin{align}
        \sum_x
        \Big(
        (M_{\Lambda_D}(\hat{\textbf{d}})+\textbf{t}).\textbf{r}(x)\Big)^2=&\sum_x
        \Big(
        \textbf{d}^{\prime}.\textbf{r}(x)\Big)^2\nonumber\\
        =&\frac{4|\textbf{d}^{\prime}|^2}{3}\nonumber\\
        =&\frac{4|(M_{\Lambda_D}(\hat{\textbf{d}})+\textbf{t})|^2}{3}
    \end{align}

    Proceeding in the same manner as in Theorem~\ref{Th:SIC_IC_Pow}, which led to Eq. \eqref{Eq:IC_power_bound}, and using Eq. \eqref{Eq:Channel_decomp}, we get
    \begin{align}
        \tilde{\mathfrak{D}}(\Lambda)=\tilde{\mathfrak{D}}(\Lambda_D)
        &\geq
        \inf_{\substack{\rho,\rho^{\prime}\\\rho\neq\rho^{\prime}}}
        \frac{
        \sum_x\Big|
        \tr[\rho \Lambda_D^{\dagger}(A^{\prime}(x))]
        -
        \tr[\rho^{\prime} \Lambda_D^{\dagger}(A^{\prime}(x))]
        \Big|
        }{
        ||\rho-\rho^{\prime}||_1
        }\nonumber\\
        &=
        \min_{\hat{|\textbf{d}|}=1}
        \frac{1}{4}
        \sum_x
        \Big|
        (M_{\Lambda_D}(\hat{\textbf{d}})+\textbf{t})
        \cdot
        \textbf{r}(x)
        \Big|
        \nonumber\\
        &\geq
        \min_{\hat{|\textbf{d}|}=1}
        \frac{
        |M_{\Lambda_D}(\hat{\textbf{d}})+\textbf{t}|
        }{
        \sqrt{6}
        }.
    \end{align}

    Now, for any arbitrary unit vector $\hat{\textbf{d}}$, the triangle inequality implies
    \begin{align}
        |M_{\Lambda_D}(\hat{\textbf{d}})+\textbf{t}|
        &\geq|M_{\Lambda_D}(\hat{\textbf{d}})|-
        |\textbf{t}|
        \nonumber\\
        &=
        \sqrt{
        \lambda_1^2d_1^2+
        \lambda_2^2d_2^2+
        \lambda_3^2d_3^2
        }
        -
        |\textbf{t}|
        \nonumber\\
        &\geq
        |\lambda_3|
        \sqrt{
        d_1^2+d_2^2+d_3^2
        }
        -
        |\textbf{t}|
        \nonumber\\
        &=
        |\lambda_3|-|\textbf{t}|.
    \end{align}

    Since $|\lambda_1|\geq|\lambda_2|\geq|\lambda_3|$, we have $\lambda_{min}=|\lambda_3|$.

    Therefore,
    \begin{align}
        \tilde{\mathfrak{D}}(\Lambda)
        \geq
        \frac{
        |\lambda_{min}|-|\textbf{t}|
        }{
        \sqrt{6}
        }.
    \end{align}

    Since $\tilde{\mathfrak{D}}(\Lambda)\geq0$ by definition, we finally obtain
    \begin{align}
        \tilde{\mathfrak{D}}(\Lambda)
        \geq
        \frac{
        \max(0,|\lambda_{min}|-|\textbf{t}|)
        }{
        \sqrt{6}
        }.
    \end{align}
\end{proof}

Combining Lemmas \ref{Lem:IC_Pres_up_Bound} and \ref{Lem:IC_pres_lower_bound}, we can conclude the following theorem regarding bounds on the IC--preservability of quantum channels.

\begin{theorem}
    For any quantum channel $\Lambda$, the IC-preservability satisfies 
    \begin{align}
        \frac{
        \max(0,|\lambda_{min}|-|\textbf{t}|)
        }{
        \sqrt{6}
        }\leq\tilde{\mathfrak{D}}(\Lambda)\leq|\lambda_{min}|.
    \end{align}\label{Th:IC_Power_Bounds}
\end{theorem}

\subsection{Absolute output coherence of Quantum Channels and its relation with IC preservability}\label{Sec:Coh_Pres_Pow}

In this section, to establish a conceptual relationship between information completeness and quantum coherence, w define another important quantity, namely the absolute output coherence of a quantum channel as: 

\begin{definition}
    Given a qubit channel $\Lambda$, we define its \emph{absolute output coherence} as
    \begin{align}
        \mathfrak{C}(\Lambda):=
        \min_{E=\{\ket{e},\ket{e^\perp}\}}
        \max_{\rho}
    \mathcal{C}_{E}(\Lambda(\rho)),
        \label{Eq:Def_Coh_Pres_Power}
    \end{align}
    where $\mathcal{C}_{E}$ denotes a valid measure of quantum coherence\cite{Stretslov_2017_Coh_Res} w.r.t. an arbitrarily chosen incoherent basis $E=\{\ket{e},\ket{e^\perp}\}$.
\end{definition}

\begin{remark}
    \rm{$\mathfrak{C}$ refers to the minimum amount of coherence (w.r.t. respect of an arbitrarily chosen incoherent basis) that can be obtained from the output of a given qubit channel for a suitable input state. This definition can be generalized for quantum channels acting on an arbitrary finite-dimensional Hilbert space. Furthermore, just like the IC-preservability of a channel, from the above definition, it follows immediately that the absolute output coherence is also invariant under unitary pre and post-processing. Explicitly,
\begin{align}
    \mathfrak{C}(\Lambda)
    =
    \mathfrak{C}(\mathcal{U}\circ\Lambda\circ\mathcal{V}),
\end{align}
where $\mathcal{U}$ and $\mathcal{V}$ are unitary channels. Now, it is easily seen that  for unitary channels, $\mathfrak{C}$ is maximum and equal to that of identity channel. It is also easy to see that $\mathfrak{C}$ is zero for the completely dephasing channel w.r.t. an arbitrary basis. Furthermore, note that $\min$ function and $\max$ function do not commute as it can be easily seen that $\max_{\rho}\min_{E=\{\ket{e},\ket{e^\perp}\}}\mathcal{C}_{E}(\Lambda(\rho))=0$ by choosing $E$ as the eigen basis of $\Lambda(\rho)$. }
\end{remark}

 Note that in Eq. \eqref{Eq:Def_Coh_Pres_Power}, the coherence measure can be arbitrarily chosen, in general. From now on we restrict ourselves to distance-based coherence measurement defined in Eq. \eqref{Eq:distance_based_coherence_measure} where the distance between two arbitrary qubit states $\rho,\rho^{\prime}$ is chosen to be $||\rho-\rho^{\prime}||_1$. We now establish the bounds on the absolute output coherence of a quantum channel.

\begin{lemma}
    For a given qubit channel $\Lambda$, its absolute output coherence is lower-bounded as
    \begin{align}
        \mathfrak{C}(\Lambda)\geq|\lambda_2|
    \end{align}
    where the signed singular values associated with the given channel $\Lambda$ satisfies the ordering $|\lambda_{max}|=|\lambda_1|\geq|\lambda_2|\geq|\lambda_3|=|\lambda_{min}|$.\label{Lem:Coh_Pres_Pow_Low_Bound}
\end{lemma}
\begin{proof}
As mentioned above, the measure of coherence~\cite{Stretslov_2015,Stretslov_2017_Coh_Res} is chosen to be
\begin{align}
    \mathcal{C}_E(\Lambda(\rho))
    =
    \min_{\rho^{\prime}_E\in\mathcal{I}_E}
    \mathcal{D}(\Lambda(\rho),\rho^{\prime}_E),
\end{align}
where $\mathcal{I}_E$ denotes the set of incoherent states with respect to the basis $E=\{\ket{e},\ket{e^\perp}\}$, i.e.,
\begin{align}
    \rho^{\prime}_E
    =
    p\ket{e}\bra{e}
    +(1-p)\ket{e^\perp}\bra{e^\perp}.
\end{align}

In this work, we choose the distance measure $\mathcal{D}$ to be the trace distance. Explicitly,
\begin{align}
    \mathcal{D}(\Lambda(\rho),\rho^{\prime}_E)
    =
    ||\Lambda(\rho)-\rho^{\prime}_E||_1.
\end{align}

Consequently, the absolute output coherence of an arbitrary quantum channel can be written as
\begin{align}
    \mathfrak{C}(\Lambda):=
    \min_{E=\{\ket{e},\ket{e^\perp}\}}
    \max_{\rho}
    \min_{\rho^{\prime}_E\in\mathcal{I}}
    ||\Lambda(\rho)-\rho^{\prime}_E||_1.
    \label{Eq:Def_Coh_Pres_Power1}
\end{align}

Using Eq. \eqref{Eq:Bloch_vec_trans_chan}, the trace distance can be expressed as
\begin{align}
    ||\Lambda(\rho)-\rho^{\prime}_E||_1
    =
    |( M_{\Lambda_D}(\textbf{s}) + \textbf{t}) - \textbf{s}^{\prime}_E|,
\end{align}
where $\textbf{s}$ and $\textbf{s}^{\prime}_E$ denote the Bloch vectors corresponding to $\rho$ and $\rho^{\prime}_E$, respectively.

Now, if
\begin{align}
   \ket{e}\bra{e}
   &=
   \frac{\mathbbm{1}+\hat{\textbf{e}}\cdot\sigma}{2}, \\
   \ket{e^\perp}\bra{e^\perp}
   &=
   \frac{\mathbbm{1}-\hat{\textbf{e}}\cdot\sigma}{2},
\end{align}
where $\hat{\textbf{e}}$ is the Bloch vector associated with the reference basis defining incoherent states, then
\begin{align}
     \rho^{\prime}_E
    &=
    p\ket{e}\bra{e}
    +(1-p)\ket{e^\perp}\bra{e^\perp}\nonumber\\
    &=
    \frac{
    \mathbbm{1}
    +(2p-1)\hat{\textbf{e}}\cdot\sigma
    }{2},
\end{align}
with $p\in[0,1]$, implying that $2p-1\in[-1,1]$. Therefore,
\begin{align}
    \textbf{s}_E=(2p-1)\hat{\textbf{e}}.
\end{align}

Hence, the set of incoherent states forms a line passing through the center of the Bloch sphere.

The absolute output coherence can therefore be rewritten as
\begin{align}
     \mathfrak{C}(\Lambda):=
    \min_{\hat{\textbf{e}}}
    \max_{\substack{\textbf{s}\\|\textbf{s}|\leq1}}
    \min_{p\in[0,1]}
    |( M_{\Lambda_D}(\textbf{s}) + \textbf{t}) - (2p-1)\hat{\textbf{e}}|.
\end{align}

We now decompose $ M_{\Lambda_D}(\textbf{s})+\textbf{t}$ into components parallel and perpendicular to $\hat{\textbf{e}}$. Explicitly,
\begin{align}
    M_{\Lambda_D}(\textbf{s})+\textbf{t}
    =
    \bigl(( M_{\Lambda_D}(\textbf{s})+\textbf{t})\cdot\hat{\textbf{e}}\bigr)\hat{\textbf{e}}
    +
    \Proj_{\hat{\textbf{e}}^\perp}( M_{\Lambda_D}(\textbf{s})+\textbf{t}),
\end{align}
where $\Proj_{\hat{\textbf{e}}^\perp}$ denotes the projection onto the plane $\hat{\textbf{e}}^\perp$, orthogonal to $\hat{\textbf{e}}$.

Consequently,
\begin{align}
    |( M_{\Lambda_D}&(\textbf{s}) + \textbf{t}) - (2p-1)\hat{\textbf{e}}|\nonumber\\
    =&
    \sqrt{
    \Bigl(
    ( M_{\Lambda_D}(\textbf{s})+\textbf{t})\cdot\hat{\textbf{e}}
    -(2p-1)
    \Bigr)^2
    +
    |\Proj_{\hat{\textbf{e}}^\perp}( M_{\Lambda_D}(\textbf{s})+\textbf{t})|^2
    }.
\end{align}

The above expression is minimized when
\begin{align}
    (2p-1)
    =
    ( M_{\Lambda_D}(\textbf{s})+\textbf{t})\cdot\hat{\textbf{e}}.
\end{align}
It is always possible to find such a $p$ as $| M_{\Lambda_D}(\textbf{s}) + \textbf{t}|\leq 1$ for all $\textbf{s}$ with $|\textbf{s}|\leq 1$ as this is the condition of $\Lambda(\rho)$ being a valid quantum state.
Therefore,
\begin{align}
    \min_{p\in[0,1]}
    |( M_{\Lambda_D}(\textbf{s}) + \textbf{t}) - (2p-1)\hat{\textbf{e}}|
    =
    |\Proj_{\hat{\textbf{e}}^\perp}( M_{\Lambda_D}(\textbf{s})+\textbf{t})|.
\end{align}

Thus, the absolute output coherence simplifies to
\begin{align}
     \mathfrak{C}(\Lambda):=
    \min_{\hat{\textbf{e}}}
    \max_{\substack{\textbf{s}\\|\textbf{s}|\leq1}}
    |\Proj_{\hat{\textbf{e}}^\perp}( M_{\Lambda_D}(\textbf{s})+\textbf{t})|.
    \label{Eq:Def_Coh_Pres_Power2}
\end{align}
  
 Since the Bloch ball is symmetric, both $\textbf{s}$ and $-\textbf{s}$ are feasible solutions of Eq. \eqref{Eq:Def_Coh_Pres_Power2}. Using the triangle inequality, we obtain
\begin{align}
2|\Proj_{\hat{\textbf{e}}^\perp}( M_{\Lambda_D}&(\textbf{s}))|\nonumber\\
=&
|\Proj_{\hat{\textbf{e}}^\perp}( M_{\Lambda_D}(\textbf{s}))
+\Proj_{\hat{\textbf{e}}^\perp}(\textbf{t})
+\Proj_{\hat{\textbf{e}}^\perp}( M_{\Lambda_D}(\textbf{s}))
-\Proj_{\hat{\textbf{e}}^\perp}(\textbf{t})|\nonumber\\
\leq&
|\Proj_{\hat{\textbf{e}}^\perp}( M_{\Lambda_D}(\textbf{s}))
+\Proj_{\hat{\textbf{e}}^\perp}(\textbf{t})|
+
|\Proj_{\hat{\textbf{e}}^\perp}( M_{\Lambda_D}(\textbf{s}))
-\Proj_{\hat{\textbf{e}}^\perp}(\textbf{t})|\nonumber\\
=&
|\Proj_{\hat{\textbf{e}}^\perp}( M_{\Lambda_D}(\textbf{s}))
+\Proj_{\hat{\textbf{e}}^\perp}(\textbf{t})|
+
|\Proj_{\hat{\textbf{e}}^\perp}(M_{\Lambda_D}(-\textbf{s}))
+\Proj_{\hat{\textbf{e}}^\perp}(\textbf{t})|.
\end{align}

Consequently,
\begin{align}
|\Proj_{\hat{\textbf{e}}^\perp}&(M_{\Lambda_D}\textbf{s})|\nonumber\\
\leq&
\max\Big(
|\Proj_{\hat{\textbf{e}}^\perp}( M_{\Lambda_D}(\textbf{s}))
+\Proj_{\hat{\textbf{e}}^\perp}(\textbf{t})|,
|\Proj_{\hat{\textbf{e}}^\perp}(M_{\Lambda_D}(-\textbf{s}))
+\Proj_{\hat{\textbf{e}}^\perp}(\textbf{t})|
\Big).
\end{align}

Maximizing over all Bloch vectors satisfying $|\textbf{s}|\leq1$, we obtain
\begin{align}
 \max_{\substack{\textbf{s}\\|\textbf{s}|\leq1}}
 |\Proj_{\hat{\textbf{e}}^\perp}( M_{\Lambda_D}(\textbf{s})+\textbf{t})|
\ge
 \max_{\substack{\textbf{s}\\|\textbf{s}|\leq1}}
 |\Proj_{\hat{\textbf{e}}^\perp}( M_{\Lambda_D}(\textbf{s}))|.
\end{align}

Therefore,
\begin{align}
\mathfrak{C}(\Lambda)
\geq
\min_{\hat{\textbf{e}}}
\max_{\substack{\textbf{s}\\|\textbf{s}|\leq1}}
|\Proj_{\hat{\textbf{e}}^\perp}( M_{\Lambda_D}(\textbf{s}))|.
\label{Eq:Def_Coh_Pres_Power3}
\end{align}

Next, consider an arbitrary plane $W$ passing through the origin of the Bloch sphere, and let $\hat{\textbf{w}}$ denote the unit vector perpendicular to it. We can always find a unit vector of the form $\hat{\textbf{x}}=(x_1,x_2,0)$ such that projection of $\hat{\textbf{w}}$ in $1-2$ plane is orthogonal to $M_{\Lambda_D}(\hat{\textbf{x}})=(\lambda_1x_1,\lambda_2x_2,0)$. Therefore,
\begin{align}
\max_{\substack{\textbf{s}\\|\textbf{s}|\leq1}}
|\Proj_{W}( M_{\Lambda_D}(\textbf{s}))|
&\geq
|\Proj_{W}(M_{\Lambda_D}(\hat{\textbf{x}}))|\nonumber\\
&=
|M_{\Lambda_D}(\hat{\textbf{x}})-(M_{\Lambda_D}(\hat{\textbf{x}}).\hat{\textbf{w}})\hat{\textbf{w}}|\nonumber\\
&=
\sqrt{|M_{\Lambda_D}(\hat{\textbf{x}})|^2-(M_{\Lambda_D}(\hat{\textbf{x}}).\hat{\textbf{w}})^2}\nonumber\\
&=
\sqrt{\lambda_1^2x_1^2+\lambda_2^2x_2^2}\nonumber\\
&\geq
|\lambda_2|
\sqrt{x_1^2+x_2^2}\nonumber\\
&=
|\lambda_2|.\label{Eq:Coh_pres_pow_bound_arb_plane}
\end{align}

Since this holds for an arbitrary plane $W$, it holds for the minimum $\hat{\textbf{e}}$ in Eq. \eqref{Eq:Def_Coh_Pres_Power3}. Let us assume that the minimum happens for some $\hat{\textbf{e}}^{\prime}$. The plane perpendicular to it is denoted by $\hat{\textbf{e}}^{\prime\perp}$. As Eq. \eqref{Eq:Coh_pres_pow_bound_arb_plane} holds for this plane also, we have
\begin{align}
    \mathfrak{C}(\Lambda)
    &\geq
    \min_{\hat{\textbf{e}}}
    \max_{\substack{\textbf{s}\\|\textbf{s}|\leq1}}
    |\Proj_{\hat{\textbf{e}}^\perp}( M_{\Lambda_D}(\textbf{s}))|\nonumber\\
    &=\max_{\substack{\textbf{s}\\|\textbf{s}|\leq1}}
    |\Proj_{\hat{\textbf{e}}^{\prime\perp}}( M_{\Lambda_D}(\textbf{s}))|\nonumber\\
    &\geq
    |\lambda_2|.
    \label{Eq:Coh_Pres_Pow_Low_Bound}
\end{align}
\end{proof}

Here, we would like to mention that for an arbitrary ordering of the modulus of singular values $\lambda_1,\lambda_2$ and $\lambda_3$, Lemma \ref{Lem:Coh_Pres_Pow_Low_Bound}
still holds i.e., it is always lower bounded by the modulus of the singular value which is neither maximum nor minimum.

As we know $|\lambda_2|\geq|\lambda_{min}|$, from Theorem \ref{Th:IC_Power_Bounds} and Lemma \ref{Lem:Coh_Pres_Pow_Low_Bound}, we obtain the following theorem that establishes the connection between the IC-preservability and the absolute output coherence of a channel.

\begin{theorem}
    For a given qubit channel $\Lambda$, its absolute output coherence is lower bounded by its IC-preservability power. That is,
    \begin{align}
        \tilde{\mathfrak{D}}(\Lambda)\leq\mathfrak{C}(\Lambda).
    \end{align}\label{Th:IC_Coh_rel}
\end{theorem}
Next, we will provide an upper bound on the absolute output coherence of a quantum channel.
\begin{lemma}
    For a given qubit channel $\Lambda$, its absolute output coherence is upper-bounded as
    \begin{align}
        \mathfrak{C}(\Lambda)\leq\max_{i\in\{1,2,3\}}|\lambda_i|=:\lambda_{max}.
    \end{align}\label{Lem:Coh_Pres_Pow_Up_Bound}
\end{lemma}

\begin{proof}
    Without loss of generality, we can assume that the signed singular values associated with the given channel $\Lambda$ satisfy the ordering $|\lambda_1|\geq|\lambda_2|\geq|\lambda_3|$. Therefore, $\lambda_{max}=\lambda_1$.
    Now, in Eq. \eqref{Eq:Def_Coh_Pres_Power2}, consider $\hat{\textbf{e}}=\frac{\textbf{t}}{|\textbf{t}|}=:\hat{\textbf{t}}$. Then, $\Proj_{\hat{\textbf{t}}^{\perp}}(\textbf{t})=0$. Thus we get
    \begin{align}
         \mathfrak{C}(\Lambda)=&
    \min_{\hat{\textbf{e}}}
    \max_{\substack{\textbf{s}\\|\textbf{s}|\leq1}}
    |\Proj_{\hat{\textbf{e}}^\perp}( M_{\Lambda_D}(\textbf{s})+\textbf{t})|\nonumber\\
    \leq&\max_{\substack{\textbf{s}\\|\textbf{s}|\leq1}}
    |\Proj_{\hat{\textbf{t}}^\perp}(M_{\Lambda_D}(\textbf{s}))|\nonumber\\
   =&
\max_{\substack{\textbf{s}\\|\textbf{s}|\leq1}}\sqrt{|M_{\Lambda_D}(\textbf{s})|^2-(M_{\Lambda_D}(\textbf{s}).\hat{\textbf{t}})^2}\nonumber\\
\leq&\max_{\substack{\textbf{s}\\|\textbf{s}|\leq1}}|M_{\Lambda_D}(\textbf{s})|\nonumber\\
=&
\max_{\substack{\textbf{s}\\|\textbf{s}|\leq1}}\sqrt{\lambda_1^2s_1^2+\lambda_2^2s_2^2+\lambda_3^2s_3^2}\nonumber\\
=&|\lambda_1|=\lambda_{max}
    \end{align}
    where the last equality occurs for $\textbf{s}=(1,0,0)^T$.
\end{proof}

Just like in the case of IC-preservability, from Lemmas \ref{Lem:Coh_Pres_Pow_Low_Bound} and \ref{Lem:Coh_Pres_Pow_Up_Bound}, we obtain the following theorem:
\begin{theorem}
    For any qubit quantum channel $\Lambda$, the absolute output coherence satisfies 
    \begin{align}
       |\lambda_2|\leq\mathfrak{C}(\Lambda)\leq|\lambda_1|,
    \end{align}\label{Th:Coh_Power_Bounds}
     where the signed singular values associated with the given channel $\Lambda$ satisfies the ordering $|\lambda_{max}|=|\lambda_1|\geq|\lambda_2|\geq|\lambda_3|= |\lambda_{min}|$.\label{Th:Coh_Power_Bounds}
\end{theorem}

If $\Lambda$ is a unital channel, absolute output coherence simplifies as follows:
\begin{proposition}
    For a given qubit unital channel $\Lambda$, the absolute output coherence simplifies as 
    \begin{align}
        \mathfrak{C}(\Lambda)=|\lambda_2|,
    \end{align}
    where the signed singular values associated with the given channel $\Lambda$ satisfies the ordering $|\lambda_{max}|=|\lambda_1|\geq|\lambda_2|\geq|\lambda_3|=|\lambda_{min}|$.\label{Prop:Coh_Pow_Unital}
\end{proposition}
\begin{proof}
    If $\Lambda$ is a unital channel, then $\textbf{t}=(0,0,0)$. Hence, absolute output coherence in Eq. \eqref{Eq:Def_Coh_Pres_Power2} simplifies to
    \begin{align}
        \mathfrak{C}(\Lambda)
=
\min_{\hat{\textbf{e}}}
\max_{\substack{\textbf{s}\\|\textbf{s}|\leq1}}
|\Proj_{\hat{\textbf{e}}^\perp}( M_{\Lambda_D}(\textbf{s}))|.
\label{Eq:Def_Coh_Pres_Power3}
    \end{align}
    Consider $\hat{\textbf{e}}=(1,0,0)$. So $\hat{\textbf{e}}^{\perp}$ is essentially the $2-3$ plane. If $\Proj_{2-3}$ denotes the projection on the $2-3$ plane then
    \begin{align}
        \mathfrak{C}(\Lambda)
=&
\min_{\hat{\textbf{e}}}
\max_{\substack{\textbf{s}\\|\textbf{s}|\leq1}}
|\Proj_{\hat{\textbf{e}}^\perp}( M_{\Lambda_D}(\textbf{s}))|\nonumber\\
\leq&\max_{\substack{\textbf{s}\\|\textbf{s}|\leq1}}
|\Proj_{2-3}( M_{\Lambda_D}(\textbf{s}))|\nonumber\\
=&
\max_{\substack{\textbf{s}\\|\textbf{s}|\leq1}}\sqrt{|M_{\Lambda_D}(\textbf{s})|^2-(M_{\Lambda_D}(\textbf{s}).\hat{\textbf{e}})^2}\nonumber\\
=&
\max_{\substack{\textbf{s}\\|\textbf{s}|\leq1}}\sqrt{\lambda_2^2s_2^2+\lambda_3^2s_3^2}\nonumber\\
=&
|\lambda_2|.
    \end{align}
    Combining this with Lemma \ref{Lem:Coh_Pres_Pow_Low_Bound}, for a unital channel, we have
\begin{align}
    \mathfrak{C}(\Lambda)=|\lambda_2|.
\end{align}
\end{proof}

Just like the Lemma \ref{Lem:Coh_Pres_Pow_Low_Bound}, Theorem \ref{Th:Coh_Power_Bounds}, and Proposition \ref{Prop:Coh_Pow_Unital} hold for any general ordering of the modulus of singular values $\lambda_1,\lambda_2$ and $\lambda_3$.

\subsection{Illustration with a two-level system interacting with bosonic thermal bath }\label{Sec:illustration}
For illustration purposes, we will calculate the above-mentioned quantities for a simple scenario where a two-level system interacts with a bosonic thermal bath under Markovian conditions. It is a very fundamental and yet important setting because it can be used to model radiative decay of a two-level atom interacting with a bath of simple harmonic oscillators at a temperature \cite {scully1997quantum}, a trapped ion transition coupled to a zero-temperature bath of optical modes in the vacuum state \cite {Leib_trapped_ion_2003}, phonons and laser light interacting with quantum dots\cite{nazir2016modelling, Vagov_Quant_Dot_2002}, a superconducting qubit interacting with environment\cite{Wallraff_2021_Circ_ED, Blais_2004_QED}. From Eq. \eqref{Eq:Opt_Mas}, the Bloch vector $\hat{\mathbf{r}}(t)=(r_1(t),r_2(t),r_3(t))$ of the two-level system at time t  satisfies\cite{breuer_2002_theory}:
\begin{align}
    \frac{dr_1(t)}{dt}=&-\frac{\gamma}{2}r_1(t)\\
    \frac{dr_2(t)}{dt}=&-\frac{\gamma}{2}r_2(t)\\
    \frac{dr_3(t)}{dt}=&-\gamma r_3(t)-\gamma_0
\end{align}
If $\{\Lambda_t\}$ is the dynamical map corresponding to this dynamics, for the initial state of the system being the ground state $\ket{g}$, from the above differential equations, we have\cite{prathik_2019_thermalization}
 \begin{align}
        M_{(\Lambda_{t})_{D}}=
        \begin{pmatrix}
        e^{-\gamma t/2} & 0 & 0 \\
        0 &  e^{-\gamma t/2} & 0 \\
        0 & 0 &  e^{-\gamma t}
        \end{pmatrix},
        \qquad
        \textbf{t}=
        \begin{pmatrix}
        0 \\
        0 \\
        \frac{\gamma_0}{\gamma}(e^{-\gamma t}-1)
        \end{pmatrix}.\label{Eq:Aff_trans_therm}
    \end{align}

    This corresponds to a special case of \emph{generalized amplitude damping channel}. We note that in this case $|\lambda_{max}|=|\lambda_1|=|\lambda_2|=e^{-\gamma t/2}$ and $|\lambda_{min}|=|\lambda_3|=e^{-\gamma t}$. So from Theorem \ref{Th:Coh_Power_Bounds} we can conclude that in this case
    \begin{align}
        \mathfrak{C}(\Lambda_t)=e^{-\gamma t/2}.
    \end{align}

    \begin{figure}[!h]
    \centering
    \includegraphics[height=140px, width =246px]{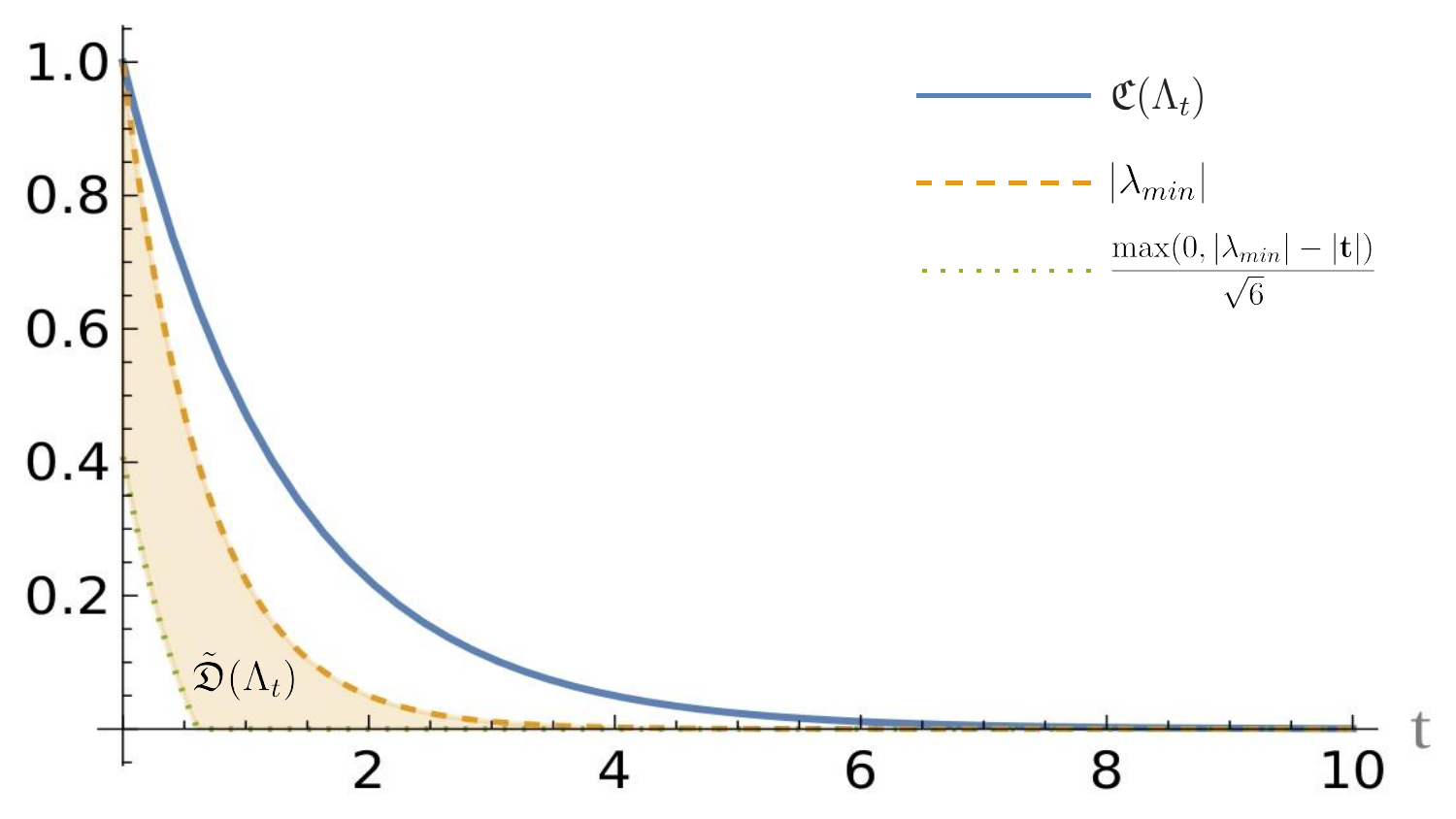}
     \caption{Setting $\gamma_0=1$ and $\gamma=1.5$, in this figure, we plot the absolute coherence output (solid curve) along with the upper (dashed curve) and the lower (dotted curve) on the IC-preservability for the channel $\Lambda$, resulting from the interaction of a two-level system with a bosonic thermal bath. The value of IC-preservability of $\Lambda_t$ lies in the shaded region.}\label {Fig:Coh_PLot} 
\end{figure}

From Fig. \ref{Fig:Coh_PLot}, we observe that $\mathfrak{C}(\Lambda_t)$ dominates the upper bound of $\tilde{\cD}(\Lambda_t)$ and both quantities monotonically tend to zero w.r.t. time. This fact shows that the system continuously loses both coherence and informational completeness w.r.t. time.

\section{Conclusion}\label{Sec:Conclusion}

In this work, we have tried to characterize and quantify the informational completeness of qubit measurements by introducing a faithful measure for it and then studying some of its properties. We have explicitly calculated its numerical value for an arbitrary qubit SIC-POVM and have demonstrated that the numerical value of this \emph{informational completeness} measure for any qubit minimal IC measurement is upper bounded by that of an SIC. Next, for an arbitrary channel, we have defined a quantity called \emph{informational completeness-preservability} that quantifies its ability to preserve informational completeness for any IC measurement when it acts on it in the Heisenberg picture. Just as in the case of \emph{informational completeness}, we proved certain useful properties of it. We then calculated bounds on the \emph{informational completeness-preservability} of a channel in terms of signed singular values associated with that quantum channel. Then, we have defined and studied \emph{absolute output coherence} of the given channel and have demonstrated that the \emph{absolute output coherence} of a qubit channel is lower bounded by its \emph{informational completeness-preservability}, thus establishing a conceptual relationship between informational completeness and quantum coherence. Finally, we have illustrated our results with examples from light-matter interaction scenarios.

Our work opens up several avenues for future research. The majority of the results in this work are proved for qubit systems, and therefore, it would be interesting to generalize the results for higher-dimensional systems. It is also important to investigate how the measure of \emph{informational completeness} relates to performance in some quantum information-theoretic tasks. It would also be interesting to study the further characterization of \emph{absolute output coherence} for higher-dimensional systems as well. 

\section{Acknowledgments}
AM acknowledges STARS (Grant No.
STARS/STARS-2/2023-0809), Government of India for support.

\bibliography{reference}

\begin{thebibliography}{42}%
\makeatletter
\providecommand \@ifxundefined [1]{%
 \@ifx{#1\undefined}
}%
\providecommand \@ifnum [1]{%
 \ifnum #1\expandafter \@firstoftwo
 \else \expandafter \@secondoftwo
 \fi
}%
\providecommand \@ifx [1]{%
 \ifx #1\expandafter \@firstoftwo
 \else \expandafter \@secondoftwo
 \fi
}%
\providecommand \natexlab [1]{#1}%
\providecommand \enquote  [1]{``#1''}%
\providecommand \bibnamefont  [1]{#1}%
\providecommand \bibfnamefont [1]{#1}%
\providecommand \citenamefont [1]{#1}%
\providecommand \href@noop [0]{\@secondoftwo}%
\providecommand \href [0]{\begingroup \@sanitize@url \@href}%
\providecommand \@href[1]{\@@startlink{#1}\@@href}%
\providecommand \@@href[1]{\endgroup#1\@@endlink}%
\providecommand \@sanitize@url [0]{\catcode `\\12\catcode `\$12\catcode `\&12\catcode `\#12\catcode `\^12\catcode `\_12\catcode `\%12\relax}%
\providecommand \@@startlink[1]{}%
\providecommand \@@endlink[0]{}%
\providecommand \url  [0]{\begingroup\@sanitize@url \@url }%
\providecommand \@url [1]{\endgroup\@href {#1}{\urlprefix }}%
\providecommand \urlprefix  [0]{URL }%
\providecommand \Eprint [0]{\href }%
\providecommand \doibase [0]{https://doi.org/}%
\providecommand \selectlanguage [0]{\@gobble}%
\providecommand \bibinfo  [0]{\@secondoftwo}%
\providecommand \bibfield  [0]{\@secondoftwo}%
\providecommand \translation [1]{[#1]}%
\providecommand \BibitemOpen [0]{}%
\providecommand \bibitemStop [0]{}%
\providecommand \bibitemNoStop [0]{.\EOS\space}%
\providecommand \EOS [0]{\spacefactor3000\relax}%
\providecommand \BibitemShut  [1]{\csname bibitem#1\endcsname}%
\let\auto@bib@innerbib\@empty
\bibitem [{\citenamefont {D’Ariano}\ \emph {et~al.}(2004)\citenamefont {D’Ariano}, \citenamefont {Perinotti},\ and\ \citenamefont {Sacchi}}]{Ariano_IC_2004}%
  \BibitemOpen
  \bibfield  {author} {\bibinfo {author} {\bibfnamefont {G.~M.}\ \bibnamefont {D’Ariano}}, \bibinfo {author} {\bibfnamefont {P.}~\bibnamefont {Perinotti}},\ and\ \bibinfo {author} {\bibfnamefont {M.~F.}\ \bibnamefont {Sacchi}},\ }\bibfield  {title} {\bibinfo {title} {Informationally complete measurements and group representation},\ }\href {https://doi.org/10.1088/1464-4266/6/6/005} {\bibfield  {journal} {\bibinfo  {journal} {Journal of Optics B: Quantum and Semiclassical Optics}\ }\textbf {\bibinfo {volume} {6}},\ \bibinfo {pages} {S487} (\bibinfo {year} {2004})}\BibitemShut {NoStop}%
\bibitem [{\citenamefont {Busch}(1991)}]{busch_1991_informationally}%
  \BibitemOpen
  \bibfield  {author} {\bibinfo {author} {\bibfnamefont {P.}~\bibnamefont {Busch}},\ }\bibfield  {title} {\bibinfo {title} {Informationally complete sets of physical quantities},\ }\href {https://link.springer.com/content/pdf/10.1007/BF00671008.pdf} {\bibfield  {journal} {\bibinfo  {journal} {International Journal of Theoretical Physics}\ }\textbf {\bibinfo {volume} {30}},\ \bibinfo {pages} {1217} (\bibinfo {year} {1991})}\BibitemShut {NoStop}%
\bibitem [{\citenamefont {D'Ariano}\ \emph {et~al.}(2005)\citenamefont {D'Ariano}, \citenamefont {Perinotti},\ and\ \citenamefont {Sacchi}}]{Ariano_2005_IC}%
  \BibitemOpen
  \bibfield  {author} {\bibinfo {author} {\bibfnamefont {G.~M.}\ \bibnamefont {D'Ariano}}, \bibinfo {author} {\bibfnamefont {P.}~\bibnamefont {Perinotti}},\ and\ \bibinfo {author} {\bibfnamefont {M.~F.}\ \bibnamefont {Sacchi}},\ }\bibfield  {title} {\bibinfo {title} {Informationally complete measurements on bipartite quantum systems: Comparing local with global measurements},\ }\href {https://doi.org/10.1103/PhysRevA.72.042108} {\bibfield  {journal} {\bibinfo  {journal} {Phys. Rev. A}\ }\textbf {\bibinfo {volume} {72}},\ \bibinfo {pages} {042108} (\bibinfo {year} {2005})}\BibitemShut {NoStop}%
\bibitem [{\citenamefont {Renes}\ \emph {et~al.}(2004)\citenamefont {Renes}, \citenamefont {Blume-Kohout},\ and\ \citenamefont {Caves}}]{Renes_SIC_Conj}%
  \BibitemOpen
  \bibfield  {author} {\bibinfo {author} {\bibfnamefont {J.~M.}\ \bibnamefont {Renes}}, \bibinfo {author} {\bibfnamefont {R.}~\bibnamefont {Blume-Kohout}},\ and\ \bibinfo {author} {\bibfnamefont {C.~M.}\ \bibnamefont {Caves}},\ }\bibfield  {title} {\bibinfo {title} {Symmetric informationally complete quantum measurements},\ }\href {https://pubs.aip.org/aip/jmp/article/45/6/2171/231563} {\bibfield  {journal} {\bibinfo  {journal} {Journal of Mathematical Physics}\ }\textbf {\bibinfo {volume} {45}},\ \bibinfo {pages} {2171} (\bibinfo {year} {2004})}\BibitemShut {NoStop}%
\bibitem [{\citenamefont {Samuel}\ and\ \citenamefont {Gedik}(2025)}]{samuel_2025_symmetric}%
  \BibitemOpen
  \bibfield  {author} {\bibinfo {author} {\bibfnamefont {S.~B.}\ \bibnamefont {Samuel}}\ and\ \bibinfo {author} {\bibfnamefont {Z.}~\bibnamefont {Gedik}},\ }\bibfield  {title} {\bibinfo {title} {Symmetric informationally complete positive operators-valued measures and the knaster’s conjecture},\ }\href {https://pubs.aip.org/aip/jmp/article/66/9/092201/3362817} {\bibfield  {journal} {\bibinfo  {journal} {Journal of Mathematical Physics}\ }\textbf {\bibinfo {volume} {66}} (\bibinfo {year} {2025})}\BibitemShut {NoStop}%
\bibitem [{\citenamefont {Wootters}\ and\ \citenamefont {Fields}(1989)}]{Wooters_1989}%
  \BibitemOpen
  \bibfield  {author} {\bibinfo {author} {\bibfnamefont {W.~K.}\ \bibnamefont {Wootters}}\ and\ \bibinfo {author} {\bibfnamefont {B.~D.}\ \bibnamefont {Fields}},\ }\bibfield  {title} {\bibinfo {title} {Optimal state-determination by mutually unbiased measurements},\ }\href {https://doi.org/https://doi.org/10.1016/0003-4916(89)90322-9} {\bibfield  {journal} {\bibinfo  {journal} {Annals of Physics}\ }\textbf {\bibinfo {volume} {191}},\ \bibinfo {pages} {363} (\bibinfo {year} {1989})}\BibitemShut {NoStop}%
\bibitem [{\citenamefont {Scott}(2006)}]{Scott_2006}%
  \BibitemOpen
  \bibfield  {author} {\bibinfo {author} {\bibfnamefont {A.~J.}\ \bibnamefont {Scott}},\ }\bibfield  {title} {\bibinfo {title} {Tight informationally complete quantum measurements},\ }\href {https://doi.org/10.1088/0305-4470/39/43/009} {\bibfield  {journal} {\bibinfo  {journal} {Journal of Physics A: Mathematical and General}\ }\textbf {\bibinfo {volume} {39}},\ \bibinfo {pages} {13507} (\bibinfo {year} {2006})}\BibitemShut {NoStop}%
\bibitem [{\citenamefont {Shang}\ \emph {et~al.}(2018)\citenamefont {Shang}, \citenamefont {Asadian}, \citenamefont {Zhu},\ and\ \citenamefont {G\"uhne}}]{Shang_2018}%
  \BibitemOpen
  \bibfield  {author} {\bibinfo {author} {\bibfnamefont {J.}~\bibnamefont {Shang}}, \bibinfo {author} {\bibfnamefont {A.}~\bibnamefont {Asadian}}, \bibinfo {author} {\bibfnamefont {H.}~\bibnamefont {Zhu}},\ and\ \bibinfo {author} {\bibfnamefont {O.}~\bibnamefont {G\"uhne}},\ }\bibfield  {title} {\bibinfo {title} {Enhanced entanglement criterion via symmetric informationally complete measurements},\ }\href {https://doi.org/10.1103/PhysRevA.98.022309} {\bibfield  {journal} {\bibinfo  {journal} {Phys. Rev. A}\ }\textbf {\bibinfo {volume} {98}},\ \bibinfo {pages} {022309} (\bibinfo {year} {2018})}\BibitemShut {NoStop}%
\bibitem [{\citenamefont {Brunner}\ \emph {et~al.}(2013)\citenamefont {Brunner}, \citenamefont {Navascu\'es},\ and\ \citenamefont {V\'ertesi}}]{Brunner_2013}%
  \BibitemOpen
  \bibfield  {author} {\bibinfo {author} {\bibfnamefont {N.}~\bibnamefont {Brunner}}, \bibinfo {author} {\bibfnamefont {M.}~\bibnamefont {Navascu\'es}},\ and\ \bibinfo {author} {\bibfnamefont {T.}~\bibnamefont {V\'ertesi}},\ }\bibfield  {title} {\bibinfo {title} {Dimension witnesses and quantum state discrimination},\ }\href {https://doi.org/10.1103/PhysRevLett.110.150501} {\bibfield  {journal} {\bibinfo  {journal} {Phys. Rev. Lett.}\ }\textbf {\bibinfo {volume} {110}},\ \bibinfo {pages} {150501} (\bibinfo {year} {2013})}\BibitemShut {NoStop}%
\bibitem [{\citenamefont {Andersson}\ \emph {et~al.}(2018)\citenamefont {Andersson}, \citenamefont {Badzi\k{a}g}, \citenamefont {Dumitru},\ and\ \citenamefont {Cabello}}]{Andersson_2018}%
  \BibitemOpen
  \bibfield  {author} {\bibinfo {author} {\bibfnamefont {O.}~\bibnamefont {Andersson}}, \bibinfo {author} {\bibfnamefont {P.}~\bibnamefont {Badzi\k{a}g}}, \bibinfo {author} {\bibfnamefont {I.}~\bibnamefont {Dumitru}},\ and\ \bibinfo {author} {\bibfnamefont {A.}~\bibnamefont {Cabello}},\ }\bibfield  {title} {\bibinfo {title} {Device-independent certification of two bits of randomness from one entangled bit and gisin's elegant bell inequality},\ }\href {https://doi.org/10.1103/PhysRevA.97.012314} {\bibfield  {journal} {\bibinfo  {journal} {Phys. Rev. A}\ }\textbf {\bibinfo {volume} {97}},\ \bibinfo {pages} {012314} (\bibinfo {year} {2018})}\BibitemShut {NoStop}%
\bibitem [{\citenamefont {Tavakoli}\ \emph {et~al.}(2020)\citenamefont {Tavakoli}, \citenamefont {Smania}, \citenamefont {Vértesi}, \citenamefont {Brunner},\ and\ \citenamefont {Bourennane}}]{Tavakoli_2020}%
  \BibitemOpen
  \bibfield  {author} {\bibinfo {author} {\bibfnamefont {A.}~\bibnamefont {Tavakoli}}, \bibinfo {author} {\bibfnamefont {M.}~\bibnamefont {Smania}}, \bibinfo {author} {\bibfnamefont {T.}~\bibnamefont {Vértesi}}, \bibinfo {author} {\bibfnamefont {N.}~\bibnamefont {Brunner}},\ and\ \bibinfo {author} {\bibfnamefont {M.}~\bibnamefont {Bourennane}},\ }\bibfield  {title} {\bibinfo {title} {Self-testing nonprojective quantum measurements in prepare-and-measure experiments},\ }\bibfield  {journal} {\bibinfo  {journal} {Science Advances}\ }\textbf {\bibinfo {volume} {6}},\ \href {https://doi.org/10.1126/sciadv.aaw6664} {10.1126/sciadv.aaw6664} (\bibinfo {year} {2020})\BibitemShut {NoStop}%
\bibitem [{\citenamefont {Mironowicz}\ and\ \citenamefont {Paw\l{}owski}(2019)}]{Miron_2019}%
  \BibitemOpen
  \bibfield  {author} {\bibinfo {author} {\bibfnamefont {P.}~\bibnamefont {Mironowicz}}\ and\ \bibinfo {author} {\bibfnamefont {M.}~\bibnamefont {Paw\l{}owski}},\ }\bibfield  {title} {\bibinfo {title} {Experimentally feasible semi-device-independent certification of four-outcome positive-operator-valued measurements},\ }\href {https://doi.org/10.1103/PhysRevA.100.030301} {\bibfield  {journal} {\bibinfo  {journal} {Phys. Rev. A}\ }\textbf {\bibinfo {volume} {100}},\ \bibinfo {pages} {030301(R)} (\bibinfo {year} {2019})}\BibitemShut {NoStop}%
\bibitem [{\citenamefont {Mitra}\ and\ \citenamefont {Ghai}(2026)}]{Mitra_2026}%
  \BibitemOpen
  \bibfield  {author} {\bibinfo {author} {\bibfnamefont {A.}~\bibnamefont {Mitra}}\ and\ \bibinfo {author} {\bibfnamefont {J.}~\bibnamefont {Ghai}},\ }\href {https://arxiv.org/abs/2512.07822} {\bibinfo {title} {Comparing quantum channels using hermitian-preserving trace-preserving linear maps: A physically meaningful approach}} (\bibinfo {year} {2026}),\ \Eprint {https://arxiv.org/abs/2512.07822} {arXiv:2512.07822 [quant-ph]} \BibitemShut {NoStop}%
\bibitem [{\citenamefont {Heinosaari}\ and\ \citenamefont {Ziman}(2011)}]{Heinosaari_book_QF}%
  \BibitemOpen
  \bibfield  {author} {\bibinfo {author} {\bibfnamefont {T.}~\bibnamefont {Heinosaari}}\ and\ \bibinfo {author} {\bibfnamefont {M.}~\bibnamefont {Ziman}},\ }\href {https://doi.org/10.1017/CBO9781139031103} {\emph {\bibinfo {title} {The mathematical language of quantum theory: from uncertainty to entanglement}}}\ (\bibinfo  {publisher} {Cambridge University Press},\ \bibinfo {year} {2011})\BibitemShut {NoStop}%
\bibitem [{\citenamefont {Watrous}(2018)}]{watrous_2018_theory}%
  \BibitemOpen
  \bibfield  {author} {\bibinfo {author} {\bibfnamefont {J.}~\bibnamefont {Watrous}},\ }\href {https://cs.uwaterloo.ca/~watrous/TQI/TQI.pdf} {\emph {\bibinfo {title} {The theory of quantum information}}}\ (\bibinfo  {publisher} {Cambridge university press},\ \bibinfo {year} {2018})\BibitemShut {NoStop}%
\bibitem [{\citenamefont {Uola}\ \emph {et~al.}(2015)\citenamefont {Uola}, \citenamefont {Budroni}, \citenamefont {G\"uhne},\ and\ \citenamefont {Pellonp\"a\"a}}]{Roope_2015_Rob}%
  \BibitemOpen
  \bibfield  {author} {\bibinfo {author} {\bibfnamefont {R.}~\bibnamefont {Uola}}, \bibinfo {author} {\bibfnamefont {C.}~\bibnamefont {Budroni}}, \bibinfo {author} {\bibfnamefont {O.}~\bibnamefont {G\"uhne}},\ and\ \bibinfo {author} {\bibfnamefont {J.-P.}\ \bibnamefont {Pellonp\"a\"a}},\ }\bibfield  {title} {\bibinfo {title} {One-to-one mapping between steering and joint measurability problems},\ }\href {https://doi.org/10.1103/PhysRevLett.115.230402} {\bibfield  {journal} {\bibinfo  {journal} {Phys. Rev. Lett.}\ }\textbf {\bibinfo {volume} {115}},\ \bibinfo {pages} {230402} (\bibinfo {year} {2015})}\BibitemShut {NoStop}%
\bibitem [{\citenamefont {Skrzypczyk}\ \emph {et~al.}(2019)\citenamefont {Skrzypczyk}, \citenamefont {\ifmmode \check{S}\else \v{S}\fi{}upi\ifmmode~\acute{c}\else \'{c}\fi{}},\ and\ \citenamefont {Cavalcanti}}]{Paul_2019_Rob}%
  \BibitemOpen
  \bibfield  {author} {\bibinfo {author} {\bibfnamefont {P.}~\bibnamefont {Skrzypczyk}}, \bibinfo {author} {\bibfnamefont {I.}~\bibnamefont {\ifmmode \check{S}\else \v{S}\fi{}upi\ifmmode~\acute{c}\else \'{c}\fi{}}},\ and\ \bibinfo {author} {\bibfnamefont {D.}~\bibnamefont {Cavalcanti}},\ }\bibfield  {title} {\bibinfo {title} {All sets of incompatible measurements give an advantage in quantum state discrimination},\ }\href {https://doi.org/10.1103/PhysRevLett.122.130403} {\bibfield  {journal} {\bibinfo  {journal} {Phys. Rev. Lett.}\ }\textbf {\bibinfo {volume} {122}},\ \bibinfo {pages} {130403} (\bibinfo {year} {2019})}\BibitemShut {NoStop}%
\bibitem [{\citenamefont {Uola}\ \emph {et~al.}(2020)\citenamefont {Uola}, \citenamefont {Bullock}, \citenamefont {Kraft}, \citenamefont {Pellonp\"a\"a},\ and\ \citenamefont {Brunner}}]{Roope_2020_Weight}%
  \BibitemOpen
  \bibfield  {author} {\bibinfo {author} {\bibfnamefont {R.}~\bibnamefont {Uola}}, \bibinfo {author} {\bibfnamefont {T.}~\bibnamefont {Bullock}}, \bibinfo {author} {\bibfnamefont {T.}~\bibnamefont {Kraft}}, \bibinfo {author} {\bibfnamefont {J.-P.}\ \bibnamefont {Pellonp\"a\"a}},\ and\ \bibinfo {author} {\bibfnamefont {N.}~\bibnamefont {Brunner}},\ }\bibfield  {title} {\bibinfo {title} {All quantum resources provide an advantage in exclusion tasks},\ }\href {https://doi.org/10.1103/PhysRevLett.125.110402} {\bibfield  {journal} {\bibinfo  {journal} {Phys. Rev. Lett.}\ }\textbf {\bibinfo {volume} {125}},\ \bibinfo {pages} {110402} (\bibinfo {year} {2020})}\BibitemShut {NoStop}%
\bibitem [{\citenamefont {Heinosaari}\ and\ \citenamefont {Miyadera}(2017)}]{Heinosaari_incomp_chan}%
  \BibitemOpen
  \bibfield  {author} {\bibinfo {author} {\bibfnamefont {T.}~\bibnamefont {Heinosaari}}\ and\ \bibinfo {author} {\bibfnamefont {T.}~\bibnamefont {Miyadera}},\ }\bibfield  {title} {\bibinfo {title} {Incompatibility of quantum channels},\ }\href {https://doi.org/10.1088/1751-8121/aa5f6b} {\bibfield  {journal} {\bibinfo  {journal} {Journal of Physics A: Mathematical and Theoretical}\ }\textbf {\bibinfo {volume} {50}},\ \bibinfo {pages} {135302} (\bibinfo {year} {2017})}\BibitemShut {NoStop}%
\bibitem [{\citenamefont {Buscemi}(2018)}]{Buscemi_2018_rev_data}%
  \BibitemOpen
  \bibfield  {author} {\bibinfo {author} {\bibfnamefont {F.}~\bibnamefont {Buscemi}},\ }\bibfield  {title} {\bibinfo {title} {Reverse data-processing theorems and computational second laws},\ }in\ \href {https://link.springer.com/chapter/10.1007/978-981-13-2487-1_6} {\emph {\bibinfo {booktitle} {Reality and Measurement in Algebraic Quantum Theory}}}\ (\bibinfo  {publisher} {Springer Singapore},\ \bibinfo {address} {Singapore},\ \bibinfo {year} {2018})\ pp.\ \bibinfo {pages} {135--159}\BibitemShut {NoStop}%
\bibitem [{\citenamefont {Buscemi}(2012)}]{Buscemi_2012_comparison}%
  \BibitemOpen
  \bibfield  {author} {\bibinfo {author} {\bibfnamefont {F.}~\bibnamefont {Buscemi}},\ }\bibfield  {title} {\bibinfo {title} {Comparison of quantum statistical models: equivalent conditions for sufficiency},\ }\href {https://link.springer.com/content/pdf/10.1007/s00220-012-1421-3.pdf} {\bibfield  {journal} {\bibinfo  {journal} {Communications in Mathematical Physics}\ }\textbf {\bibinfo {volume} {310}},\ \bibinfo {pages} {625} (\bibinfo {year} {2012})}\BibitemShut {NoStop}%
\bibitem [{\citenamefont {Buscemi}(2016)}]{Buscemi_2016_degradable}%
  \BibitemOpen
  \bibfield  {author} {\bibinfo {author} {\bibfnamefont {F.}~\bibnamefont {Buscemi}},\ }\bibfield  {title} {\bibinfo {title} {Degradable channels, less noisy channels, and quantum statistical morphisms: An equivalence relation},\ }\href {https://link.springer.com/content/pdf/10.1134/S0032946016030017.pdf} {\bibfield  {journal} {\bibinfo  {journal} {Problems of Information Transmission}\ }\textbf {\bibinfo {volume} {52}},\ \bibinfo {pages} {201} (\bibinfo {year} {2016})}\BibitemShut {NoStop}%
\bibitem [{\citenamefont {Yashin}\ \emph {et~al.}(2020)\citenamefont {Yashin}, \citenamefont {Kiktenko}, \citenamefont {Mastiukova},\ and\ \citenamefont {Fedorov}}]{yashin_2020_minimal}%
  \BibitemOpen
  \bibfield  {author} {\bibinfo {author} {\bibfnamefont {V.}~\bibnamefont {Yashin}}, \bibinfo {author} {\bibfnamefont {E.}~\bibnamefont {Kiktenko}}, \bibinfo {author} {\bibfnamefont {A.}~\bibnamefont {Mastiukova}},\ and\ \bibinfo {author} {\bibfnamefont {A.}~\bibnamefont {Fedorov}},\ }\bibfield  {title} {\bibinfo {title} {Minimal informationally complete measurements for probability representation of quantum dynamics},\ }\href {https://iopscience.iop.org/article/10.1088/1367-2630/abb963/pdf} {\bibfield  {journal} {\bibinfo  {journal} {New Journal of Physics}\ }\textbf {\bibinfo {volume} {22}},\ \bibinfo {pages} {103026} (\bibinfo {year} {2020})}\BibitemShut {NoStop}%
\bibitem [{\citenamefont {Appleby}\ \emph {et~al.}(2017)\citenamefont {Appleby}, \citenamefont {Flammia}, \citenamefont {McConnell},\ and\ \citenamefont {Yard}}]{Appleby_SIC_2017}%
  \BibitemOpen
  \bibfield  {author} {\bibinfo {author} {\bibfnamefont {M.}~\bibnamefont {Appleby}}, \bibinfo {author} {\bibfnamefont {S.}~\bibnamefont {Flammia}}, \bibinfo {author} {\bibfnamefont {G.}~\bibnamefont {McConnell}},\ and\ \bibinfo {author} {\bibfnamefont {J.}~\bibnamefont {Yard}},\ }\bibfield  {title} {\bibinfo {title} {Sics and algebraic number theory},\ }\href {https://link.springer.com/article/10.1007/s10701-017-0090-7} {\bibfield  {journal} {\bibinfo  {journal} {Foundations of Physics}\ }\textbf {\bibinfo {volume} {47}},\ \bibinfo {pages} {1042} (\bibinfo {year} {2017})}\BibitemShut {NoStop}%
\bibitem [{\citenamefont {Appleby}\ \emph {et~al.}(2018)\citenamefont {Appleby}, \citenamefont {Chien}, \citenamefont {Flammia},\ and\ \citenamefont {Waldron}}]{Appleby_SIC_2018}%
  \BibitemOpen
  \bibfield  {author} {\bibinfo {author} {\bibfnamefont {M.}~\bibnamefont {Appleby}}, \bibinfo {author} {\bibfnamefont {T.-Y.}\ \bibnamefont {Chien}}, \bibinfo {author} {\bibfnamefont {S.}~\bibnamefont {Flammia}},\ and\ \bibinfo {author} {\bibfnamefont {S.}~\bibnamefont {Waldron}},\ }\bibfield  {title} {\bibinfo {title} {Constructing exact symmetric informationally complete measurements from numerical solutions},\ }\href {https://doi.org/10.1088/1751-8121/aab4cd} {\bibfield  {journal} {\bibinfo  {journal} {Journal of Physics A: Mathematical and Theoretical}\ }\textbf {\bibinfo {volume} {51}},\ \bibinfo {pages} {165302} (\bibinfo {year} {2018})}\BibitemShut {NoStop}%
\bibitem [{\citenamefont {Appleby}\ \emph {et~al.}(2022)\citenamefont {Appleby}, \citenamefont {Bengtsson}, \citenamefont {Grassl}, \citenamefont {Harrison},\ and\ \citenamefont {McConnell}}]{Appleby_SIC_2022}%
  \BibitemOpen
  \bibfield  {author} {\bibinfo {author} {\bibfnamefont {M.}~\bibnamefont {Appleby}}, \bibinfo {author} {\bibfnamefont {I.}~\bibnamefont {Bengtsson}}, \bibinfo {author} {\bibfnamefont {M.}~\bibnamefont {Grassl}}, \bibinfo {author} {\bibfnamefont {M.}~\bibnamefont {Harrison}},\ and\ \bibinfo {author} {\bibfnamefont {G.}~\bibnamefont {McConnell}},\ }\bibfield  {title} {\bibinfo {title} {Sic-povms from stark units: Prime dimensions n2+ 3},\ }\href {https://pubs.aip.org/aip/jmp/article/63/11/112205/2846122} {\bibfield  {journal} {\bibinfo  {journal} {Journal of Mathematical Physics}\ }\textbf {\bibinfo {volume} {63}} (\bibinfo {year} {2022})}\BibitemShut {NoStop}%
\bibitem [{\citenamefont {Bengtsson}\ \emph {et~al.}(2025)\citenamefont {Bengtsson}, \citenamefont {Grassl},\ and\ \citenamefont {McConnell}}]{Bengtsson_SIC_2025}%
  \BibitemOpen
  \bibfield  {author} {\bibinfo {author} {\bibfnamefont {I.}~\bibnamefont {Bengtsson}}, \bibinfo {author} {\bibfnamefont {M.}~\bibnamefont {Grassl}},\ and\ \bibinfo {author} {\bibfnamefont {G.}~\bibnamefont {McConnell}},\ }\bibfield  {title} {\bibinfo {title} {Sic-povms from stark units: Dimensions n2+ 3= 4p, p prime},\ }\href {https://pubs.aip.org/aip/jmp/article/66/8/082202/3358437} {\bibfield  {journal} {\bibinfo  {journal} {Journal of Mathematical Physics}\ }\textbf {\bibinfo {volume} {66}} (\bibinfo {year} {2025})}\BibitemShut {NoStop}%
\bibitem [{\citenamefont {Scott}(2017)}]{Scott_SIC_num}%
  \BibitemOpen
  \bibfield  {author} {\bibinfo {author} {\bibfnamefont {A.~J.}\ \bibnamefont {Scott}},\ }\href {https://arxiv.org/abs/1703.03993} {\bibinfo {title} {Sics: Extending the list of solutions}} (\bibinfo {year} {2017}),\ \Eprint {https://arxiv.org/abs/1703.03993} {arXiv:1703.03993 [quant-ph]} \BibitemShut {NoStop}%
\bibitem [{\citenamefont {Braun}\ \emph {et~al.}(2014)\citenamefont {Braun}, \citenamefont {Giraud}, \citenamefont {Nechita}, \citenamefont {Pellegrini},\ and\ \citenamefont {{\v{Z}}nidari{\v{c}}}}]{braun_2014_universal}%
  \BibitemOpen
  \bibfield  {author} {\bibinfo {author} {\bibfnamefont {D.}~\bibnamefont {Braun}}, \bibinfo {author} {\bibfnamefont {O.}~\bibnamefont {Giraud}}, \bibinfo {author} {\bibfnamefont {I.}~\bibnamefont {Nechita}}, \bibinfo {author} {\bibfnamefont {C.}~\bibnamefont {Pellegrini}},\ and\ \bibinfo {author} {\bibfnamefont {M.}~\bibnamefont {{\v{Z}}nidari{\v{c}}}},\ }\bibfield  {title} {\bibinfo {title} {A universal set of qubit quantum channels},\ }\href {https://iopscience.iop.org/article/10.1088/1751-8113/47/13/135302/pdf} {\bibfield  {journal} {\bibinfo  {journal} {Journal of Physics A: Mathematical and Theoretical}\ }\textbf {\bibinfo {volume} {47}},\ \bibinfo {pages} {135302} (\bibinfo {year} {2014})}\BibitemShut {NoStop}%
\bibitem [{\citenamefont {King}\ and\ \citenamefont {Ruskai}(2001)}]{Ruskai_2001}%
  \BibitemOpen
  \bibfield  {author} {\bibinfo {author} {\bibfnamefont {C.}~\bibnamefont {King}}\ and\ \bibinfo {author} {\bibfnamefont {M.}~\bibnamefont {Ruskai}},\ }\bibfield  {title} {\bibinfo {title} {Minimal entropy of states emerging from noisy quantum channels},\ }\href {https://doi.org/10.1109/18.904522} {\bibfield  {journal} {\bibinfo  {journal} {IEEE Transactions on Information Theory}\ }\textbf {\bibinfo {volume} {47}},\ \bibinfo {pages} {192} (\bibinfo {year} {2001})}\BibitemShut {NoStop}%
\bibitem [{\citenamefont {{Beth Ruskai}}\ \emph {et~al.}(2002)\citenamefont {{Beth Ruskai}}, \citenamefont {Szarek},\ and\ \citenamefont {Werner}}]{Ruskai_2002}%
  \BibitemOpen
  \bibfield  {author} {\bibinfo {author} {\bibfnamefont {M.}~\bibnamefont {{Beth Ruskai}}}, \bibinfo {author} {\bibfnamefont {S.}~\bibnamefont {Szarek}},\ and\ \bibinfo {author} {\bibfnamefont {E.}~\bibnamefont {Werner}},\ }\bibfield  {title} {\bibinfo {title} {An analysis of completely-positive trace-preserving maps on m2},\ }\href {https://doi.org/https://doi.org/10.1016/S0024-3795(01)00547-X} {\bibfield  {journal} {\bibinfo  {journal} {Linear Algebra and its Applications}\ }\textbf {\bibinfo {volume} {347}},\ \bibinfo {pages} {159} (\bibinfo {year} {2002})}\BibitemShut {NoStop}%
\bibitem [{\citenamefont {Streltsov}\ \emph {et~al.}(2015)\citenamefont {Streltsov}, \citenamefont {Singh}, \citenamefont {Dhar}, \citenamefont {Bera},\ and\ \citenamefont {Adesso}}]{Stretslov_2015}%
  \BibitemOpen
  \bibfield  {author} {\bibinfo {author} {\bibfnamefont {A.}~\bibnamefont {Streltsov}}, \bibinfo {author} {\bibfnamefont {U.}~\bibnamefont {Singh}}, \bibinfo {author} {\bibfnamefont {H.~S.}\ \bibnamefont {Dhar}}, \bibinfo {author} {\bibfnamefont {M.~N.}\ \bibnamefont {Bera}},\ and\ \bibinfo {author} {\bibfnamefont {G.}~\bibnamefont {Adesso}},\ }\bibfield  {title} {\bibinfo {title} {Measuring quantum coherence with entanglement},\ }\href {https://doi.org/10.1103/PhysRevLett.115.020403} {\bibfield  {journal} {\bibinfo  {journal} {Phys. Rev. Lett.}\ }\textbf {\bibinfo {volume} {115}},\ \bibinfo {pages} {020403} (\bibinfo {year} {2015})}\BibitemShut {NoStop}%
\bibitem [{\citenamefont {Baumgratz}\ \emph {et~al.}(2014)\citenamefont {Baumgratz}, \citenamefont {Cramer},\ and\ \citenamefont {Plenio}}]{Baum_2014}%
  \BibitemOpen
  \bibfield  {author} {\bibinfo {author} {\bibfnamefont {T.}~\bibnamefont {Baumgratz}}, \bibinfo {author} {\bibfnamefont {M.}~\bibnamefont {Cramer}},\ and\ \bibinfo {author} {\bibfnamefont {M.~B.}\ \bibnamefont {Plenio}},\ }\bibfield  {title} {\bibinfo {title} {Quantifying coherence},\ }\href {https://doi.org/10.1103/PhysRevLett.113.140401} {\bibfield  {journal} {\bibinfo  {journal} {Phys. Rev. Lett.}\ }\textbf {\bibinfo {volume} {113}},\ \bibinfo {pages} {140401} (\bibinfo {year} {2014})}\BibitemShut {NoStop}%
\bibitem [{\citenamefont {Streltsov}\ \emph {et~al.}(2017)\citenamefont {Streltsov}, \citenamefont {Adesso},\ and\ \citenamefont {Plenio}}]{Stretslov_2017_Coh_Res}%
  \BibitemOpen
  \bibfield  {author} {\bibinfo {author} {\bibfnamefont {A.}~\bibnamefont {Streltsov}}, \bibinfo {author} {\bibfnamefont {G.}~\bibnamefont {Adesso}},\ and\ \bibinfo {author} {\bibfnamefont {M.~B.}\ \bibnamefont {Plenio}},\ }\bibfield  {title} {\bibinfo {title} {Colloquium: Quantum coherence as a resource},\ }\href {https://doi.org/10.1103/RevModPhys.89.041003} {\bibfield  {journal} {\bibinfo  {journal} {Rev. Mod. Phys.}\ }\textbf {\bibinfo {volume} {89}},\ \bibinfo {pages} {041003} (\bibinfo {year} {2017})}\BibitemShut {NoStop}%
\bibitem [{\citenamefont {Breuer}\ and\ \citenamefont {Petruccione}(2002)}]{breuer_2002_theory}%
  \BibitemOpen
  \bibfield  {author} {\bibinfo {author} {\bibfnamefont {H.-P.}\ \bibnamefont {Breuer}}\ and\ \bibinfo {author} {\bibfnamefont {F.}~\bibnamefont {Petruccione}},\ }\href {https://books.google.com/books?hl=en&lr=&id=0Yx5VzaMYm8C&oi=fnd&pg=PA8&dq=Breuer+Fettuccine+theory+of+open+quantum+systems&ots=RV3nF-lD3E&sig=E2INy0zVe4CHf8DuAGWZPc57nRc} {\emph {\bibinfo {title} {The theory of open quantum systems}}}\ (\bibinfo  {publisher} {OUP Oxford},\ \bibinfo {year} {2002})\BibitemShut {NoStop}%
\bibitem [{\citenamefont {Scully}\ and\ \citenamefont {Zubairy}(1997)}]{scully1997quantum}%
  \BibitemOpen
  \bibfield  {author} {\bibinfo {author} {\bibfnamefont {M.~O.}\ \bibnamefont {Scully}}\ and\ \bibinfo {author} {\bibfnamefont {M.~S.}\ \bibnamefont {Zubairy}},\ }\href@noop {} {\emph {\bibinfo {title} {Quantum optics}}}\ (\bibinfo  {publisher} {Cambridge university press},\ \bibinfo {year} {1997})\BibitemShut {NoStop}%
\bibitem [{\citenamefont {Leibfried}\ \emph {et~al.}(2003)\citenamefont {Leibfried}, \citenamefont {Blatt}, \citenamefont {Monroe},\ and\ \citenamefont {Wineland}}]{Leib_trapped_ion_2003}%
  \BibitemOpen
  \bibfield  {author} {\bibinfo {author} {\bibfnamefont {D.}~\bibnamefont {Leibfried}}, \bibinfo {author} {\bibfnamefont {R.}~\bibnamefont {Blatt}}, \bibinfo {author} {\bibfnamefont {C.}~\bibnamefont {Monroe}},\ and\ \bibinfo {author} {\bibfnamefont {D.}~\bibnamefont {Wineland}},\ }\bibfield  {title} {\bibinfo {title} {Quantum dynamics of single trapped ions},\ }\href {https://doi.org/10.1103/RevModPhys.75.281} {\bibfield  {journal} {\bibinfo  {journal} {Rev. Mod. Phys.}\ }\textbf {\bibinfo {volume} {75}},\ \bibinfo {pages} {281} (\bibinfo {year} {2003})}\BibitemShut {NoStop}%
\bibitem [{\citenamefont {Nazir}\ and\ \citenamefont {McCutcheon}(2016)}]{nazir2016modelling}%
  \BibitemOpen
  \bibfield  {author} {\bibinfo {author} {\bibfnamefont {A.}~\bibnamefont {Nazir}}\ and\ \bibinfo {author} {\bibfnamefont {D.~P.}\ \bibnamefont {McCutcheon}},\ }\bibfield  {title} {\bibinfo {title} {Modelling exciton--phonon interactions in optically driven quantum dots},\ }\href {https://iopscience.iop.org/article/10.1088/0953-8984/28/10/103002/pdf} {\bibfield  {journal} {\bibinfo  {journal} {Journal of Physics: Condensed Matter}\ }\textbf {\bibinfo {volume} {28}},\ \bibinfo {pages} {103002} (\bibinfo {year} {2016})}\BibitemShut {NoStop}%
\bibitem [{\citenamefont {Vagov}\ \emph {et~al.}(2002)\citenamefont {Vagov}, \citenamefont {Axt},\ and\ \citenamefont {Kuhn}}]{Vagov_Quant_Dot_2002}%
  \BibitemOpen
  \bibfield  {author} {\bibinfo {author} {\bibfnamefont {A.}~\bibnamefont {Vagov}}, \bibinfo {author} {\bibfnamefont {V.~M.}\ \bibnamefont {Axt}},\ and\ \bibinfo {author} {\bibfnamefont {T.}~\bibnamefont {Kuhn}},\ }\bibfield  {title} {\bibinfo {title} {Electron-phonon dynamics in optically excited quantum dots: Exact solution for multiple ultrashort laser pulses},\ }\href {https://doi.org/10.1103/PhysRevB.66.165312} {\bibfield  {journal} {\bibinfo  {journal} {Phys. Rev. B}\ }\textbf {\bibinfo {volume} {66}},\ \bibinfo {pages} {165312} (\bibinfo {year} {2002})}\BibitemShut {NoStop}%
\bibitem [{\citenamefont {Blais}\ \emph {et~al.}(2021)\citenamefont {Blais}, \citenamefont {Grimsmo}, \citenamefont {Girvin},\ and\ \citenamefont {Wallraff}}]{Wallraff_2021_Circ_ED}%
  \BibitemOpen
  \bibfield  {author} {\bibinfo {author} {\bibfnamefont {A.}~\bibnamefont {Blais}}, \bibinfo {author} {\bibfnamefont {A.~L.}\ \bibnamefont {Grimsmo}}, \bibinfo {author} {\bibfnamefont {S.~M.}\ \bibnamefont {Girvin}},\ and\ \bibinfo {author} {\bibfnamefont {A.}~\bibnamefont {Wallraff}},\ }\bibfield  {title} {\bibinfo {title} {Circuit quantum electrodynamics},\ }\href {https://doi.org/10.1103/RevModPhys.93.025005} {\bibfield  {journal} {\bibinfo  {journal} {Rev. Mod. Phys.}\ }\textbf {\bibinfo {volume} {93}},\ \bibinfo {pages} {025005} (\bibinfo {year} {2021})}\BibitemShut {NoStop}%
\bibitem [{\citenamefont {Blais}\ \emph {et~al.}(2004)\citenamefont {Blais}, \citenamefont {Huang}, \citenamefont {Wallraff}, \citenamefont {Girvin},\ and\ \citenamefont {Schoelkopf}}]{Blais_2004_QED}%
  \BibitemOpen
  \bibfield  {author} {\bibinfo {author} {\bibfnamefont {A.}~\bibnamefont {Blais}}, \bibinfo {author} {\bibfnamefont {R.-S.}\ \bibnamefont {Huang}}, \bibinfo {author} {\bibfnamefont {A.}~\bibnamefont {Wallraff}}, \bibinfo {author} {\bibfnamefont {S.~M.}\ \bibnamefont {Girvin}},\ and\ \bibinfo {author} {\bibfnamefont {R.~J.}\ \bibnamefont {Schoelkopf}},\ }\bibfield  {title} {\bibinfo {title} {Cavity quantum electrodynamics for superconducting electrical circuits: An architecture for quantum computation},\ }\href {https://doi.org/10.1103/PhysRevA.69.062320} {\bibfield  {journal} {\bibinfo  {journal} {Phys. Rev. A}\ }\textbf {\bibinfo {volume} {69}},\ \bibinfo {pages} {062320} (\bibinfo {year} {2004})}\BibitemShut {NoStop}%
\bibitem [{\citenamefont {Prathik~Cherian}\ \emph {et~al.}(2019)\citenamefont {Prathik~Cherian}, \citenamefont {Chakraborty},\ and\ \citenamefont {Ghosh}}]{prathik_2019_thermalization}%
  \BibitemOpen
  \bibfield  {author} {\bibinfo {author} {\bibfnamefont {J.}~\bibnamefont {Prathik~Cherian}}, \bibinfo {author} {\bibfnamefont {S.}~\bibnamefont {Chakraborty}},\ and\ \bibinfo {author} {\bibfnamefont {S.}~\bibnamefont {Ghosh}},\ }\bibfield  {title} {\bibinfo {title} {On thermalization of two-level quantum systems},\ }\href {https://iopscience.iop.org/article/10.1209/0295-5075/126/40003/pdf} {\bibfield  {journal} {\bibinfo  {journal} {Europhysics Letters}\ }\textbf {\bibinfo {volume} {126}},\ \bibinfo {pages} {40003} (\bibinfo {year} {2019})}\BibitemShut {NoStop}%
\end{thebibliography}%

\end{document}